\documentclass[12pt, draftclsnofoot, onecolumn]{IEEEtran}
\usepackage{cite}
\usepackage{amsmath,amssymb,amsfonts}
\usepackage{algorithmic}
\usepackage{amsthm}
\usepackage{graphicx}
\usepackage{subfigure}
\usepackage{textcomp}
\usepackage{xcolor}
\usepackage{multirow}
\usepackage{threeparttable}
\usepackage{booktabs}
\usepackage{epstopdf}
\usepackage{stmaryrd}
\usepackage{bm}
\usepackage{algorithm}
\usepackage{algorithmic}
\newtheorem{theorem}{Theorem}
\newtheorem{lemma}[theorem]{Lemma}
\newtheorem{corollary}[theorem]{Corollary}

\newtheorem{definition}{Definition}
\newtheorem*{remark}{Remark}
\newtheoremstyle{noparens}
  {}{}
  {\itshape}{}
  {\bfseries}{.}
  { }
  {\thmname{#1}\thmnumber{ #2}\mdseries\thmnote{ #3}}
\theoremstyle{noparens}

\newtheorem{lemmaNoParens}[theorem]{Lemma}
\newtheorem{propositionNoParens}[theorem]{Proposition}
\newtheorem{definitionNoParens}[definition]{Definition}
\begin{document}

\title{On the Message Passing Efficiency of Polar and Low-Density Parity-Check Decoders}


\author{\IEEEauthorblockN{Dawei Yin,
Yuan Li,
Xianbin Wang,
Jiajie Tong,
Huazi Zhang,
Jun Wang,\\
Guanghui Wang,
Jun Chen,
Guiying Yan,
Zhiming Ma,
and Wen Tong}
\thanks{This work is partially supported by the National Natural Science Foundation of China (12231018), and presented in part at the the 2022 IEEE Globecom: Workshop on Channel Coding beyond 5G.

Dawei Yin is with the School of Mathematics, Shandong University, Jinan, 250100 China, and also with the Huawei Technologies Co. Ltd., Hangzhou, 310051 China (e-mail: daweiyin@mail.sdu.edu.cn).

Yuan Li is with the Academy of Mathematics and Systems Science, CAS, University of Chinese Academy of Sciences, Beijing, 100190 China, and also with the Huawei Technologies Co. Ltd., Hangzhou, 310051 China (e-mail: liyuan181@mails.ucas.ac.cn).

Xianbin Wang, Jiajie Tong, Huazi Zhang, and Jun Wang are with the Huawei Technologies Co. Ltd., Hangzhou, 310051 China (e-mail: \{wangxianbin1, tongjiajie, zhanghuazi, justin.wangjun\}@huawei.com).

Guanghui Wang is with the School of Mathematics, Shandong University, Jinan, 250100 China (e-mail: ghwang@sdu.edu.cn).

Wen Tong is with the Huawei Technologies Canada Co. Ltd., Ottawa, ON L3R 5A4, Canada (e-mail:tongwen@huawei.com).

Jun Chen is with the Department of Electrical and Computer Engineering, McMaster University, Hamilton, ON L8S 4K1, Canada (e-mail: chenjun@mcmaster.ca).

Guiying Yan and Zhiming Ma are with the Academy of Mathematics and Systems Science, CAS, University of Chinese Academy of Sciences, Beijing, 100190 China (e-mail: yangy@amss.ac.cn; mazm@amt.ac.cn).

}
}

\maketitle

\begin{abstract}
This study focuses on the efficiency of message-passing-based decoding algorithms for polar and low-density parity-check (LDPC) codes. Both successive cancellation (SC) and belief propagation (BP) decoding algorithms are studied {in} the message-passing framework. {Counter-intuitively, SC decoding demonstrates the highest decoding efficiency, although it was considered a weak decoder {in terms of} error-correction performance.} We analyze the complexity-performance tradeoff to dynamically track the decoding efficiency, where the complexity is measured by the number of messages passed (NMP), and the performance is measured by the statistical distance to the maximum \emph{a posteriori} (MAP) estimate. This study offers a new insight into the contribution of each message passed in decoding, and compares various decoding algorithms on a message-by-message level. The analysis corroborates recent results on terabits-per-second polar SC decoders, and might shed light on better scheduling strategies.

\end{abstract}


\IEEEdisplaynontitleabstractindextext
\IEEEpeerreviewmaketitle
\section{Introduction}
\IEEEPARstart{A}{FTER} over six decades of research efforts, probabilistic coding and polar coding can approach the theoretical Shannon limit under additive white Gaussian noise (AWGN) channel, both in the asymptotic and finite-length regimes \cite{4595172,910578,6589171,polyanskiy2010channel}. Shannon's paradigm focuses on two dimensions:  \emph{probability of decoding error} versus  \emph{noise level}. The former is measured by bit error rate (BER) or block error rate (BLER), and the latter is measured by signal-to-noise ratio (SNR). The ``SNR-BLER'' curve and its gap to channel capacity have become the primary criterion for comparing different coding and decoding schemes. This paradigm has guided coding theory and practice toward a tremendous success. In practice, there are other dimensions of interest, such as complexity, latency, and energy efficiency. In the 1960s, Berlekamp noted that \emph{``from a practical standpoint, the essential limitation of all coding and decoding schemes proposed to date has not been Shannon's capacity but the complexity of the decoder''} \cite{1968Algebraic}. Although Berlekamp refered to algebraic codes, his observation holds to this day. Complex decoding algorithms can hardly be adopted for implementation even if they approach channel capacity. This is particularly true as  Moore's law has almost reached the physical limits \cite{Morre_law}.


\begin{figure}[htbp]
    \centering
    \includegraphics[width=0.7\textwidth]{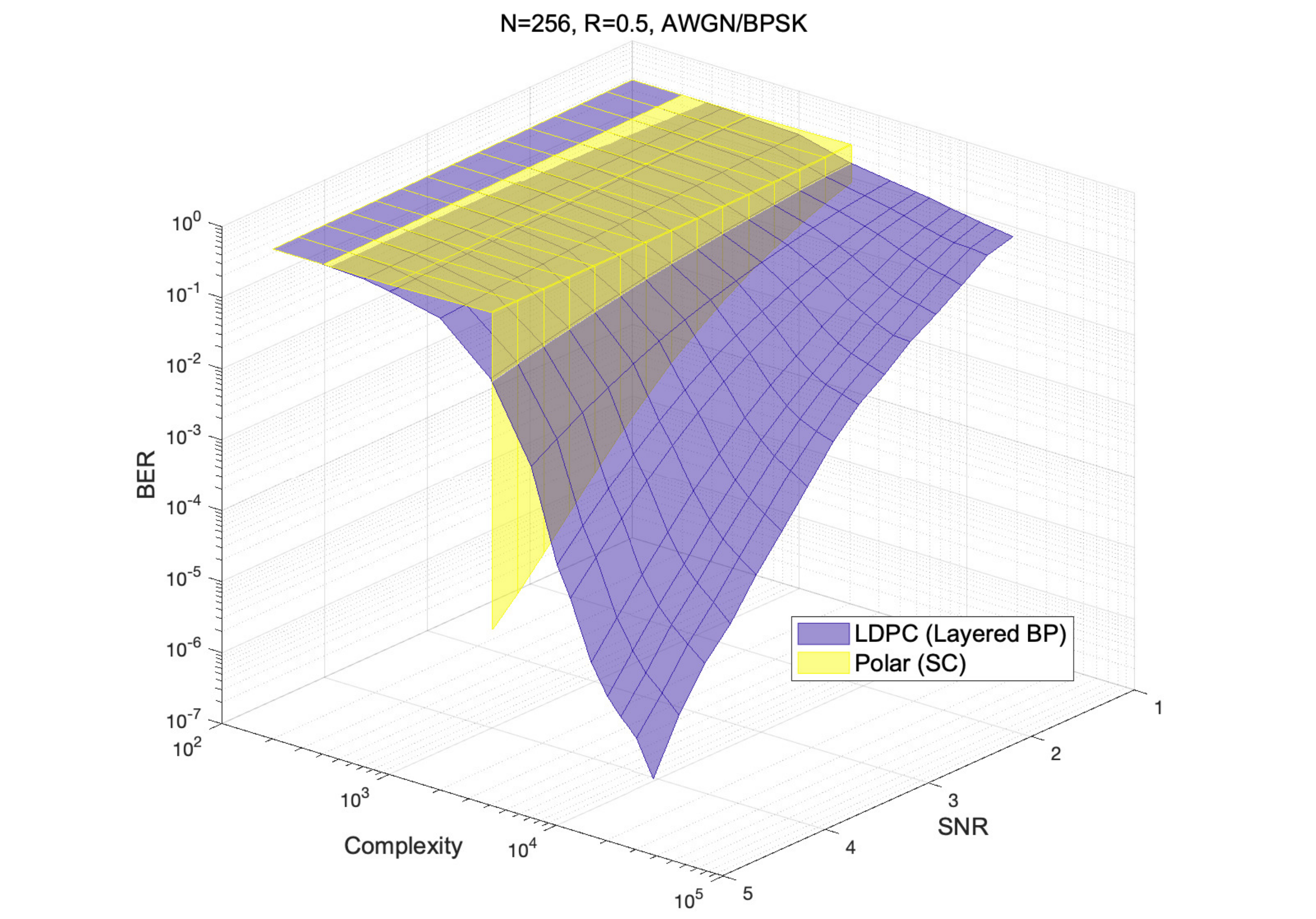}
    \caption{A 3-dimensional comparison of polar and LDPC codes (constructions follow 5G New Radio).}
    \label{fig.3d}
\end{figure}

Recently, terabits-per-second (Tbps) high throughput decoders were reported from both academia and industry \cite{8880815,sural2020tb,arxiv.2107.08600}. Polar SC decoders can achieve a record-breaking throughput of 1.5 Tbps \cite{sural2020tb} and 4 Tbps \cite{arxiv.2107.08600} within only one square millimeter chip area. These results prove the high decoding efficiency of polar SC decoders, although they were considered weak from the error-correction performance perspective \cite{5205860, 5592698,7055304, 6297420}. High-efficiency decoders are instrumental in delivering Tbps peak throughput services for 6G. It prompts us to explore the theoretical analysis of decoding efficiency, in addition to the well-researched decoding performance. In particular, we need to define a metric to analyze decoding efficiency quantitatively. {As depicted in Fig. \ref{fig.3d}, when complexity is included as a new dimension (the measure of complexity will be provided shortly), it provides a new perspective on coding schemes. In the framework of SNR-BER, it is not evident that there is a tenfold difference in complexity behind the performance.}

Decoding efficiency, characterized by a complexity-performance tradeoff, varies across different decoding algorithms. In the research of sequential decoding, the cutoff rate has become an essential index for describing the complexity-performance tradeoff \cite{1053909}. The average per-bit complexity is a finite constant when {operating} below the cutoff rate and grows exponentially if the code rate is above the cutoff rate \cite{ITRC1968}. Massey suggested that, as a rule of thumb, the cutoff rate is the practical upper limit on code rates for reliable communications, whereas capacity is the theoretical upper limit \cite{7446999}. In message-passing decoding, McEliece conjectured that for a large class of channels, if the designed rate of code ensemble equals a fraction $1-\epsilon$ of the channel capacity, then the decoding complexity scales like $\frac{1}{\epsilon}ln\frac{1}{\epsilon}$ \cite{10539999}. The decoding efficiency in circuits, that is, the energy efficiency has been studied in the very large scale integration (VLSI) model \cite{6940280,6284015,7083760, 8438511}. Blake and Kschischang established the asymptotic energy bound of LDPC and polar decoders in \cite{8719012, 7862885}. Additionally, Grover examined the optimization problem of the total energy consumption for LDPC codes, which includes both transmission and decoding energy \cite{4595173,6483282}. All these works highlighted the importance of taking decoding complexity into consideration, in addition to performance.


This study aims at measuring the decoding efficiency quantitatively such that various message-passing algorithms can be compared. We focus on two well-known codes, that is, polar codes and low-density parity-check (LDPC) codes \cite{1057683}. Specifically, we focus on successive cancellation (SC) decoding for polar codes and belief propagation (BP) decoding for both codes. In the general sense, both SC and BP are message-passing algorithms on a factor graph. To characterize the complexity-performance tradeoff, we adopt a two-dimensional paradigm, where the complexity is measured by the number of messages passed (NMP), and the performance is measured by a metric based on the Bethe free energy (to be described shortly). {The fine-grained definition of NMP allows to measure BP decoding complexity on a message-by-message level, as opposed to a coarse-grained complexity analysis on the iteration level \cite{5407615}.} The NMP can also measure the complexity of SC decoding for polar codes, as SC is BP with special scheduling \cite{7154407}. Measuring decoding performance by BER/BLER is a common practice. However, it offers little insight into the internal decoding process, thus does not serve our purpose. Alternatively, we propose to look into the convergence of decoding process by tracking \emph{``how far away is an intermediate decoding state from the maximum \emph{a posteriori} (MAP) decoding result''}, coined as ``Gap to maximum A Posteriori (GAP)''. To this end, we derive GAP based on Bethe free energy to measure the statistical distance between the distribution obtained by message passing and the true posterior distribution. It offers two distinctive advantages. First, it zooms in on each individual edge's log-likelihood ratio (LLR) update in the decoding process, thus providing a higher resolution than traditional extrinsic information transfer (EXIT) chart methods for analyzing iterative decoders. Second, the ``NMP-GAP'' curve can be efficiently obtained by Monte Carlo simulations or density evolution (DE) \cite{910578}. This paradigm allows us to quantitatively study the decoding efficiency of message-passing algorithms and may inspire new approaches for improving efficiency.


The paper is organized as follows: Section II introduces some basic concepts and several decoding algorithms. In Section III, we illustrate the rationality of adopting the ``NMP-GAP'' curve to study decoding efficiency. In Section IV, we discuss the advantages of the ``NMP-GAP'' paradigm and apply it to compare the efficiency of several algorithms. {In Section V, we apply this model to the scheduling policy of LDPC codes.} In Section VI, we discuss limitations and weaknesses in this model. We conclude in Section VII.

\section{PRELIMINARIES}
In this section, we review the BP and SC algorithms from the perspective of decoding efficiency.

\subsection{LDPC codes and the BP decoding}
{A binary LDPC code with a parity check matrix $H$ can be represented by a factor graph $G=$ $(\mathcal{V} \cup \mathcal{C}, \mathcal{E})$\cite{1056404}, where $\mathcal{V}$ and $\mathcal{C}$ are the disjoint sets of variable nodes (VNs) and check nodes (CNs), respectively, and $\mathcal{E}$ is the set of edges connecting the VNs and CNs.}

BP decoding proceeds by successively passing messages on the factor graph between VNs (denoted by $i$) and CNs (denoted by $\alpha$). Examples of LDPC and polar factor graphs are shown in Fig.~\ref{fig.factor}.

\begin{figure}[htbp]
    \centering
    \includegraphics[width=0.6\textwidth]{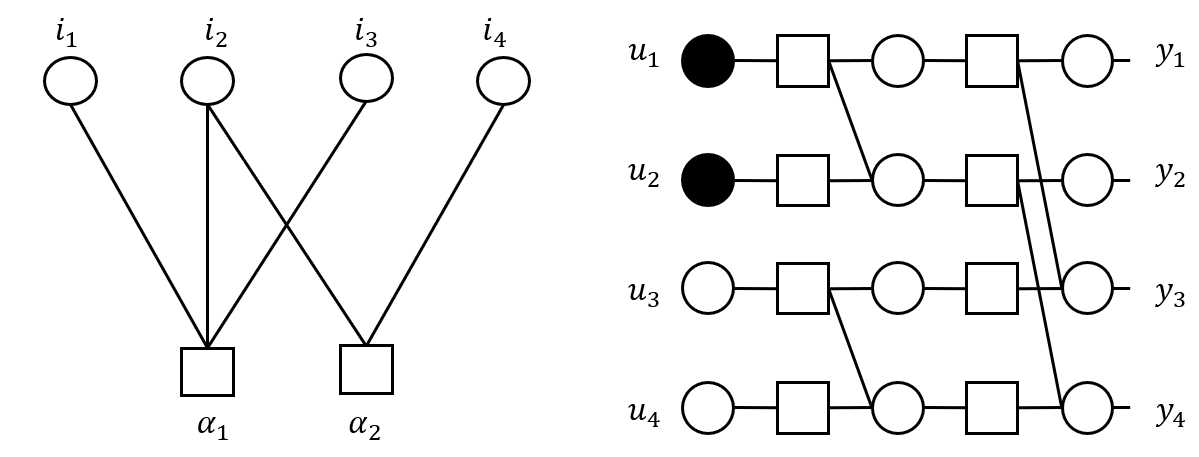}
    \caption{(4,2) LDPC and Polar codes. Let circles denote variable nodes, boxes denote check nodes, and solid circles be the frozen bits of polar code.}
    \label{fig.factor}
\end{figure}

For the sum-product algorithm \cite{910572}, the update rule for a variable-to-check (V2C) LLR message is
\begin{equation}
\label{equ:equ1}
m_{i \rightarrow \alpha}^{(l)}=m_{0}+\sum_{h \in N(i)\backslash \alpha} m_{h \rightarrow i}^{(l-1)},
\end{equation}
and the update rule at a check-to-variable (C2V) edge is
\begin{equation}
\label{equ:equ2}
    \begin{split}
m_{\alpha \rightarrow i}^{(l)}=2 \tanh ^{-1}\left(\prod_{j \in N(\alpha)\backslash i} \tanh \left(m_{j \rightarrow \alpha}^{(l-1)} / 2\right)\right),
    \end{split}
\end{equation}
where $m_{0}$ is the channel message in LLR form, $l$ is the number of iterations, $N(v)$ represents the nodes connected directly to node $v$, {$i \rightarrow \alpha$ means from variable node $i$ to check node $\alpha$ and $\alpha \rightarrow i$ means from check node $\alpha$ to variable node $i$. The initial messages $m_{i \rightarrow \alpha}^{(0)}$ and $m_{\alpha \rightarrow i}^{(0)}$ are $0$.}

The min-sum approximation in \cite{1495850} can be employed to reduce the complexity of (\ref{equ:equ2}):
\begin{equation}
\label{equ:equ3}
    \begin{split}
m_{\alpha \rightarrow i}^{(l)} \approx \prod_{j \in N(\alpha)\backslash i} \operatorname{sign}\left(m_{j \rightarrow \alpha}^{(l-1)} \right) \cdot \min _{j \in N(\alpha)\backslash i} \left(\left| m_{j \rightarrow \alpha}^{(l-1)}\right| \right).
    \end{split}
\end{equation}

\subsection{Polar codes and the SC decoding}
{Polar codes {are parameterized} by $\left(n, k, \mathcal{A}\right)$, where $n$ is the code length and $\mathcal{A}$ is a set of information indices that carry $k$ information bits $\mathbf{u}_{\mathcal{A}}$. The complement of $\mathcal{A}$ is the frozen indices $\mathcal{A}^{c}$ that carry frozen bits $\mathbf{u}_{\mathcal{A}^{c}}$. The polar generator matrix is $\mathbf{G}_{n}=\mathbf{B}_{n} \mathbf{F}^{\otimes s}$ for any $n=2^{s}$, where $\mathbf{B}_{n}$ is a bit-reversal permutation matrix, $\mathbf{F}^{\otimes s}$ denotes the $s-$th Kronecker power of $\mathbf{F} \triangleq\left[\begin{array}{c}1~0 \\ 1~1\end{array}\right]$. A codeword is generated by $\mathbf{c}=\mathbf{u} {G}_{n}$, where $\mathbf{u}=\left(\mathbf{u}_{\mathcal{A}}, \mathbf{u}_{\mathcal{A}^{c}}\right)$.

SC decoder for polar code estimates $u_{i}$ for given frozen bits $\mathbf{u}_{\mathcal{A}^{c}}$, received word $\mathbf{y}$ and the estimates $\widehat{\mathbf{u}}_{1}^{i-1}$ of $\mathbf{u}_{1}^{i-1}$ \cite{4595172}. It can be implemented by computing LLR $L_{n}^{(i)}\left(\mathbf{y}_{1}^{n}, \widehat{\mathbf{u}}_{1}^{i-1}\right) \triangleq \ln \frac{W_{n}^{(i)}\left(\mathbf{y}_{1}^{n} \widehat{\mathbf{u}}_{1}^{i-1} \mid 0\right)}{W_{n}^{(i)}\left(\mathbf{y}_{1}^{n}, \widehat{\mathbf{u}}_{1}^{i-1} \mid 1\right)}(1 \leq i \leq n)$ according to the recursive formula:
\begin{equation}
\label{equ:equ4}
L_{n}^{(2 j-1)}\left(\mathbf{y}_{1}^{n}, \widehat{\mathbf{u}}_{1}^{2 j-2}\right) \\
=F\left(L_{n / 2}^{(j)}\left(\mathbf{y}_{1}^{n / 2}, \widehat{\mathbf{u}}_{1, o}^{2 j-2} \oplus \widehat{\mathbf{u}}_{1, e}^{2 j-2}\right), L_{n / 2}^{(j)}\left(\mathbf{y}_{n / 2+1}^{n}, \widehat{\mathbf{u}}_{1, e}^{2 j-2}\right)\right),
\end{equation}

and

\begin{equation}
\label{equ:equ5}
L_{n}^{(2 j)}\left(\mathbf{y}_{1}^{n}, \widehat{\mathbf{u}}_{1}^{2 j-1}\right) \\
=G\left(L_{n / 2}^{(j)}\left(\mathbf{y}_{1}^{n / 2}, \widehat{\mathbf{u}}_{1, o}^{2 j-2} \oplus \widehat{\mathbf{u}}_{1, e}^{2 j-2}\right), L_{n / 2}^{(j)}\left(\mathbf{y}_{n / 2+1}^{n}, \widehat{\mathbf{u}}_{1, e}^{2 j-2}\right), \widehat{u}_{2 j-1}\right)
\end{equation}
for $1 \leq j \leq n / 2$, where
$$F(x, y)=2 \tanh ^{-1}(\tanh (x / 2) \tanh (y / 2)),
$$

$$
G(x, y, u) \triangleq(-1)^{u} x + y,
$$
$\widehat{\mathbf{u}}_{1, o}^{i}$ and $\widehat{\mathbf{u}}_{1, e}^{i}$ are subvectors of $\widehat{\mathbf{u}}_{1}^{i}$ with odd and even indices respectively. $L_{1}^{(1)}\left(r_{i}\right) \triangleq \ln \frac{W\left(r_{i} \mid 0\right)}{W\left(r_{i} \mid 1\right)}$ is the initial channel metric and $\hat{u}_{2j-1}$ is the modulo-2 partial sum of previously decoded bits.
}

Similar to BP, the min-sum approximation can also be employed to reduce the complexity of (\ref{equ:equ4}):
\begin{equation}
\label{equ:equ6}
    \begin{split}
F(x, y) \approx \operatorname{sign}(x) \operatorname{sign}(y) \min (|x|,|y|).
    \end{split}
\end{equation}

\subsection{Density Evolution}
The BP process can be analyzed by density evolution. Assume that $\bm{x}=\{x_1,\cdots,x_n\} \in \{0,1\}^n$ are transmitted through binary memoryless symmetric (BMS) channels $\bm{W}=\{W_1,\cdots,W_n\}$, and $\bm{y}=\{y_1,\cdots,y_n\}$ are received signals. Let
$$L(y_i)=\ln\frac{p(x_i=0|y_i)}{p(x_i=1|y_i)}$$
denote the corresponding LLR of $y_i$, and $c_i$ be the density of $L(y_i)$ conditioned on $x_i = 0$. We call $c_i$ the $L$-density of $y_i$. Then $c_i \in \mathcal{X}$ is a symmetric probability measure, where $\mathcal{X}$ is a convex subset of symmetric probability measures \cite{2008Modern}. The $L$-densities of $m_{\alpha \rightarrow i}$ and $m_{i \rightarrow \alpha}$ are denoted by $c_{(\alpha, i)}$ and $c_{(i, \alpha)}$, respectively. The entropy function of $L$-density is the linear functional defined by
$$H(c)\triangleq\int_{-\infty}^{+\infty}c(y) {log_2}\left(1+e^{-y}\right)dy.$$

The convolution operations on the variable node and check node are denoted by two binary operators $\circledast$ and $\boxast$, respectively.
{
For $L$-density $c_1$, $c_2$, and any Borel set $E \subset \overline{\mathbb{R}}$, define
$$
\left(\mathrm{c}_{1} \circledast \mathrm{c}_{2}\right)(E) \triangleq \int_{\overline{\mathbb{R}}} \mathrm{c}_{1}(E-\alpha) \mathrm{c}_{2}(d \alpha), \\
$$

$$
\left(\mathrm{c}_{1} \boxast \mathrm{c}_{2}\right)(E) \triangleq \int_{\overline{\mathbb{R}}} \mathrm{c}_{1}\left(2 \tanh ^{-1}\left(\frac{\tanh \left(\frac{E}{2}\right)}{\tanh \left(\frac{\alpha}{2}\right)}\right)\right) \mathrm{c}_{2}(d \alpha),
$$
where $\int_{\overline{\mathbb{R}}} f(\alpha) c(d \alpha)$ is the Lebesgue integral with respect to probability measure $c$ on extended real numbers $\overline{\mathbb{R}}$.

}

If the factor graph is cycle-free, then the update rule in (1), (2) can be written in the form of $L$-density
$$c_{(i, \alpha)}^{(l)}=c_i \circledast \left(\circledast_{h \in N(i)\backslash \alpha} \ c_{(h, i)}^{(l-1)}\right),$$
$$c_{(\alpha, i)}^{(l)}=\boxast_{j \in N(\alpha)\backslash i} \ c_{(j, \alpha)}^{(l-1)}.$$

{In order to analyze how the messages evolve in the decoding process, we introduce the following definitions and properties about $L$-density in \cite{6912949}.}

\begin{definitionNoParens}
{
For $L$-density $c$  and $f:[0,1] \rightarrow \mathbb{R}$, define
$$
I_{f}(c) \triangleq \int_{\overline{\mathbb{R}}} f\left(\left|\tanh \left(\frac{\alpha}{2}\right)\right|\right) c(d \alpha).
$$
We call $c_{1}$ is degraded with respect to $c_{2}$ (denoted by $c_{1} \succeq c_{2}$), if $I_{f}\left(\mathrm{c}_{1}\right) \geq I_{f}\left(\mathrm{c}_{2}\right)$ for all concave non-increasing $f$. Furthermore, $c_{1}$ is said to be strictly degraded with respect to $c_{2}$ (denoted by $ c_{1} \succ c_{2}$) if $c_{1} \succeq c_{2}$ and $c_{1} \neq c_{2}$.}
\end{definitionNoParens}

Degradation defines a partial order on the space of symmetric probability measures, with the greatest element $\Delta_{0}$ and the least element $\Delta_{\infty}$. Thus we have
$$
c \succ \Delta_{\infty} \text { if } c \neq \Delta_{\infty}, \text { and } c \prec \Delta_{0} \text { if } c \neq \Delta_{0}.
$$

The following proposition will be used in subsequent discussions.
\begin{propositionNoParens}
Suppose $\mathrm{x}_{1}, \mathrm{x}_{2}, \mathrm{x}_{3}, \mathrm{x}_{1}^{\prime}, \mathrm{x}_{2}^{\prime} \in \mathcal{X}$.

i) If $\mathrm{x}_{1} \succeq \mathrm{x}_{2}$, then
$$
\begin{aligned}
\mathrm{x}_{1} \circledast \mathrm{x}_{3} \succeq \mathrm{x}_{2} \circledast \mathrm{x}_{3}, \quad \text { for all } \mathrm{x}_{3} \in \mathcal{X},
\\
\mathrm{x}_{1} \boxast \mathrm{x}_{3} \succeq \mathrm{x}_{2} \boxast \mathrm{x}_{3}, \quad \text { for all } \mathrm{x}_{3} \in \mathcal{X}.
\end{aligned}
$$

ii) If $\mathrm{x}_{1} \succ \mathrm{x}_{2}$, then
$$H\left(\mathrm{x}_{1}\right)>H\left(\mathrm{x}_{2}\right).$$

iii)
$$H(\mathrm{x}_{1} \circledast \mathrm{x}_{2})+H(\mathrm{x}_{1} \boxast \mathrm{x}_{2})=H(\mathrm{x}_{1})+H(\mathrm{x}_{2}). $$

iv)
If $\mathrm{y}_{1}=\mathrm{x}_{1}^{\prime}-\mathrm{x}_{1}, \mathrm{y}_{2}=\mathrm{x}_{2}^{\prime}-\mathrm{x}_{2}$ with $\mathrm{x}_{1}^{\prime} \succeq \mathrm{x}_{1}, \mathrm{x}_{2}^{\prime} \succeq \mathrm{x}_{2}$, then

$$H\left(\mathrm{y}_{1} \boxast \mathrm{y}_{2}\right) \leq 0, \quad H\left(\mathrm{y}_{1} \circledast \mathrm{y}_{2}\right) \geq 0.$$

\end{propositionNoParens}

More properties about $L$-density and the entropy of symmetric measures can be found in \cite{6912949} and {chapter 4} of \cite{2008Modern}.

\subsection{{Bethe free energy}}

The Bethe free energy is an important concept in statistical physics, which is a thermodynamic quantity that provides a measure of the stability of a physical system \cite{bethe1935statistical}. The use of Bethe free energy within the context of graphical models is a critical step in theoretically understanding the BP decoding algorithms \cite{1459044,6788554}. This subsection introduces the concept of Bethe free energy {based on which subsequent variations will be developed.} 

Let $i\in [1,n]$ and $\alpha \in [1,m]$ index the variable nodes and the check nodes in a factor graph, respectively. Assume that $\bm{x}=\{x_1,\cdots,x_n\} \in \{0,1\}^n$ is transmitted through BMS channel $\bm{W}=\{W_1,\cdots,W_n\}$, and $\bm{y}=\{y_1,\cdots,y_n\}$ is the received signal. The MAP decoding is based on the following conditional density function
$$
p(x_1,...,x_n|y_1,...,y_n)=\frac{\prod_{1}^{n}{f_i(x_i)}\prod_{1}^{m}{f_\alpha(\bm{x}_\alpha)}}{Z},
$$
where $f_i(x_i)=W_i(y_i|x_i)$ is the channel transition probability, $f_\alpha(\bm{x}_\alpha) = I(\oplus \bm{x}_\alpha=0)$ {indicates} whether check node $\alpha$ holds, $\bm{x}_{\alpha}=\{ x_i,i \in N(\alpha) \}$, and $Z$ is the partition function.

Since the complexity of MAP decoding through $ \underset{\bm{x}\in C}{\arg\max}~p(x_1,...,x_n|y_1,...,y_n)$ is exponential, we need to find a density function to approximate $p$. {Chapter 14 of \cite{Montanari2009Information}} shows that if the factor graph is cycle-free, then we have the exact formula
 $$
p(\bm{x})=\frac{\prod_{\alpha}{p_\alpha(\bm{x}_\alpha)}}{\prod_{i}{p_i(x_i)}^{d_i-1}},
$$
where $d_i$ is the degree of variable node $i$, and $p_i$, $p_{\alpha}$ are marginal density functions. Hence we consider
$$
b \in \mathfrak{M}=\left\{\frac{\prod_{\alpha}{b_\alpha(\bm{x}_\alpha)}}{\prod_{i}{b_i(x_i)}^{d_i-1}} \ | \ b_i(x_i)=\sum_{\bm{x}_{\alpha} \backslash x_i} b_{\alpha}(\bm{x}_{\alpha})\right\}.
$$

The Bethe free energy
$$
F(b)\triangleq U_{bethe}(b)-H_{bethe}(b) =-\ln Z+D(b||p)
$$
characterizes the distance between $b$ and $p$, where
$$
\begin{aligned}
U_{bethe}(b)=&-\sum_{\alpha=1}^{m} \sum_{\bm{x}_{\alpha}} b_{\alpha}\left(\bm{x}_{\alpha}\right) \ln f_{\alpha}\left(\bm{x}_{\alpha}\right) -\sum_{i=1}^{n}\sum_{x_i} b_{i}\left(x_{i}\right) \ln f_{i}\left(x_{i}\right), \\
\end{aligned}
$$

$$
\begin{aligned}
H_{bethe}(b)=&-\sum_{\alpha=1}^{m} \sum_{\bm{x}_{\alpha}} b_{\alpha}\left(\bm{x}_{\alpha}\right) \ln b_{\alpha}\left(\bm{x}_{\alpha}\right) +\sum_{i=1}^{n}\left(d_{i}-1\right) \sum_{x_{i}} b_{i}\left(x_{i}\right) \ln b_{i}\left(x_{i}\right),
\end{aligned}
$$
and
$$
D(b \| p)=\sum_{\bm{x} \in \{0,1\}^n} b(\bm{x}) \log \frac{b(\bm{x})}{p(\bm{x})}
$$
is the relative entropy (Kullback-Leibler divergence). Hence it can be used to measure the gap between the density obtained by BP decoding and the MAP decoding result. For a fixed channel and coding scheme, i.e., fixed $p$, ``$-\ln Z$'' is a constant, which enables us to observe $D(b||p)$ through Bethe free energy.

By the Lagrange multiplier method, we {introduce} the Lagrange multipliers $\gamma_{\alpha}$ and $\gamma_{i}$ for the normalization constraints, $\lambda_{(\alpha i)}(x_i)$ for the marginalization constraints to find the stationary point of $F(b)$. We thus construct a Lagrangian of the form below
$$
\begin{aligned}
L=& F +\sum_{\alpha=1}^{m} \gamma_{\alpha}\left[\sum_{\boldsymbol{x}_{\alpha}} b_{\alpha}\left(\boldsymbol{x}_{\alpha}\right)-1\right]+\sum_{i=1}^{n} \gamma_{i}\left[\sum_{x_{i}} b_{i}\left(x_{i}\right)-1\right] \\
&+\sum_{i=1}^{n} \sum_{\alpha \in N(i)} \sum_{x_{i}} \lambda_{\alpha i}\left(x_{i}\right)\left[b_{i}\left(x_{i}\right)-\sum_{\bm{x}_{\alpha} \backslash x_{i}} b_{\alpha}\left(\bm{x}_{\alpha}\right)\right].
\end{aligned}
$$

Setting the derivatives of the Lagrangian with respect to the Lagrange multipliers equal to zero gives back the equality constraints. Setting the derivatives of the Lagrangian with respect to $\{b_{\alpha}(x_{\alpha}), b_{i}(x_i)\}$ equal to zero gives the equations
$$
\hat{b}_{\alpha}\left(\bm{x}_{\alpha}\right) \propto f_{\alpha}\left(\bm{x}_{\alpha}\right) \exp \left[\sum_{i \in N(\alpha)} \lambda_{\alpha i}\left(x_{i}\right)\right],
$$
and
$$
\hat{b}_{i}\left(x_{i}\right) \propto \exp \left[\frac{1}{d_{i}-1}\left(-\ln f_i(x_i) + \sum_{a \in N(i)} \lambda_{a i}\left(x_{i}\right)\right)\right] .
$$
Let
\begin{equation}
\label{equ:equ7}
\lambda_{\alpha i}\left(x_{i}\right)=\ln p_{i \rightarrow \alpha}\left(x_{i}\right)=\ln \left(f_i(x_i) \prod_{h \in N(i)\backslash \alpha} p_{h \rightarrow i}\left(x_{i}\right)\right),
\end{equation}
then we have
\begin{equation}
\label{equ:equ8}
\hat{b}_{\alpha}\left(\boldsymbol{x}_{\alpha}\right) \propto f_{\alpha}\left(\boldsymbol{x}_{\alpha}\right) \prod_{i \in N(\alpha)} p_{i \rightarrow \alpha}\left(x_{i}\right),
\end{equation}
and
\begin{equation}
\label{equ:equ9}
\hat{b}_{i}\left(x_{i}\right) \propto f_i(x_i)\prod_{\alpha \in N(i)} p_{\alpha \rightarrow i}\left(x_{i}\right).
\end{equation}

We obtain the fixed-point equations of the BP algorithm from (\ref{equ:equ8}), (\ref{equ:equ9}), and the marginalization and normalization constraints \cite{1459044}. Then we can change the variables from $\{ b_{\alpha}(\bm{x}_{\alpha}), b_i(x_i)) \}$ to $\underline{p}= \{ p_{i \rightarrow \alpha}, p_{\alpha \rightarrow i} \}$ at the BP fixed point,
$$
F_{*}(\underline{p})= -\left(\sum_{\alpha} \mathbb{F}_{\alpha}(\underline{p})+\sum_{i} \mathbb{F}_{i}(\underline{p})-\sum_{(i, \alpha)} \mathbb{F}_{i \alpha}(\underline{p})\right),
$$
where
$$
\mathbb{F}_{\alpha}(\underline{p})=\ln \left[\sum_{\bm{x}_\alpha} f_{\alpha}\left(\bm{x}_{\alpha}\right) \prod_{j \in N(\alpha)} p_{j \rightarrow \alpha}\left(x_{j}\right)\right], \quad
$$

$$
\mathbb{F}_{i}(\underline{p})=\ln \left[\sum_{x_{i}} f_i(x_i)\prod_{h \in N(i)} p_{h \rightarrow i}\left(x_{i}\right)\right], \\
$$

$$
\mathbb{F}_{i \alpha}(\underline{p})=\ln \left[\sum_{x_{i}} p_{i \rightarrow \alpha}\left(x_{i}\right) p_{\alpha \rightarrow i}\left(x_{i}\right)\right].
$$

Since the channel is BMS, we {can assume without loss of generality that the all-zero codeword is sent} and take the expectation over the channel noises, {yielding} the expectation of Bethe free energy with respect to the channel output in a cycle-free factor graph,

\begin{equation}
\label{equ:equ101}
\begin{aligned}
E_{\bm{y}}(F_{*}(\underline{p})) &=-\{\sum_{i} \mathrm{H}\left(c_{i} \circledast \left(\circledast_{\alpha \in N(i)} c_{(\alpha, i)}\right)\right) +\sum_{\alpha} \sum_{i \in N(\alpha)} \mathrm{H}\left(c_{(i, \alpha)}\right)\\
&-\sum_{\alpha} \mathrm{H}\left(\boxast_{i \in N(\alpha)} \mathrm{c}_{(i, \alpha)}\right)
-\sum_{(i, \alpha)} \mathrm{H}\left(\mathrm{c}_{(i, \alpha)} \circledast c_{(\alpha, i)}\right)\}+\sum_{i}H(c_i),
\end{aligned}
\end{equation}
where $\bm{c}=(c_1,\cdots,c_n)$ is the $L$-density of channel message, $c_{(\alpha, i)}$ and $c_{(i, \alpha)}$ are the $L$-densities of  $m_{\alpha \rightarrow i}$ and $m_{i \rightarrow \alpha}$ on the BP fixed point.

{The expectation of $E_{\bm{y}}(F_{*}(\underline{p}))$ with respect to the code ensemble (factor graphs with the same degree distribution) is the potential functional of LDPC ensemble in \cite{6912949}. Correspondingly, our study will focus on the factor graph of a specific code in the next section.}

\section{characterization of decoding efficiency}

To {conduct} an efficiency analysis, we need to explore the internal decoding process rather than focusing on the final decoding result only. Similar to the SNR-BLER paradigm, we characterize the decoding efficiency along two dimensions. Instead of the SNR-BLER curve, we adopt the \emph{number of messages passed (NMP)} and \emph{Gap to maximum A Posteriori (GAP)} to measure complexity and performance, respectively.

\subsection{Number of messages passed (NMP) as the complexity metric (x-axis)}
There are many ways to measure decoding complexity. The common ones are as follows. (1) The computational complexity in terms of the number of mathematical operations is the most commonly used. But it differs significantly across implementations \cite{8719012,7446995,7862885}. (2) Tree search based algorithms, such as Fano decoding, consider the number of nodes visited. But this metric is difficult to be generalized to other algorithms. (3) For belief-propagation-based algorithms, the number of messages passed serves our purpose. First, it applies to a wide range of polar and LDPC decoding algorithms; second, the complexity of different algorithms can be aligned; third, it is independent of hardware implementation.

In \cite{5407615}, the number of messages passed {is approximated by} the number of edges in the factor graph {times} the number of iterations.  If we only {are only concerned with} the decoding of LDPC codes, this calculation method would be sufficient. However, if we want to have a fine-grained observation of the decoding process, we need a counting method that zooms in on each edge. The number of messages passed in this work is counted as follows. {We define {passing a} V2C or C2V message on a directed edge as a unit of complexity. In the min-sum algorithm, {passing} C2V and V2C messages have similar computational complexity. We can get all the updated C2V messages of a check node by $2d_c-3$ comparison operations, where $d_c$ is the check node degree. The average number of comparison operations required for any C2V message passing is less than 2.} This property {is also satisfied by V2C message passing} in the V2C message. We add up all the LLR values passed to the variable node and subtract the corresponding message when passing to the check node.

For LDPC codes, the arrowheaded line from $c_1$ to $v_1$ in Fig. \ref{fig1}(a) is an example of C2V message. We can obtain all the LLRs from $c_1$ to adjacent nodes by finding the two smallest values of $|m_{v_i \rightarrow c_1}| \left(i\in \left(1, 4 \right)\right)$, and count the amount of floating-point computation required for each edge as shown in the first row of Table~I. V2C message passing is shown in Fig. \ref{fig1}(b), and its complexity can be analyzed in a similar way.

For polar codes, the factor graph is equivalent to a bipartite graph as shown in Fig. 2(c,d). Hence the same rules to measure complexity can be applied to the C2V or V2C update in polar codes. Specifically, the \emph{F} function in (\ref{equ:equ4}) is a C2V update of a degree-3 check node, and the \emph{G} function in (\ref{equ:equ5}) is a V2C update of a degree-2 variable node. The required computation is shown in the last two rows of Table~I. Notice that we do not {take into account} hard messages in the polar code factor graph when considering messages passed. This is in accordance with the characteristics of the SC decoding algorithm. Since the complexity of a Boolean or assignment operation is negligible compared to a floating-point operation, we can assume that a hard decision removes a node (and its incident edges) from the factor graph. This complexity model is more consistent with the actual cost. At the same time, if we run the soft cancellation (SCAN) or BP algorithm on the polar factor graph, each edge updating that involves a floating-point operation will still be counted as a message passed.

\begin{figure}[htbp]
\centering
\subfigure[C2V in LDPC code]{
\includegraphics[width=0.3\textwidth]{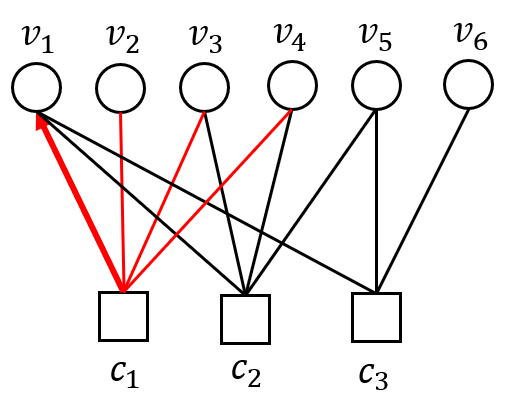}
}
\quad
\subfigure[V2C in LDPC code]{
\includegraphics[width=0.3\textwidth]{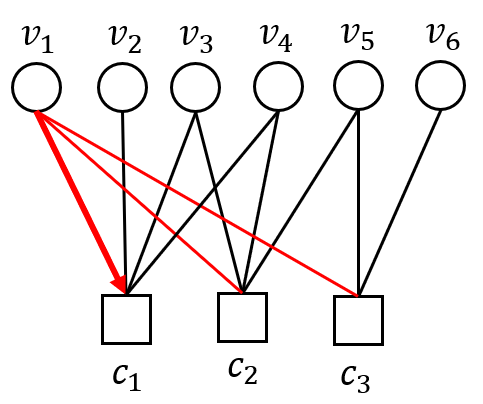}
}
\quad
\subfigure[C2V in polar code]{
\includegraphics[width=0.35\textwidth]{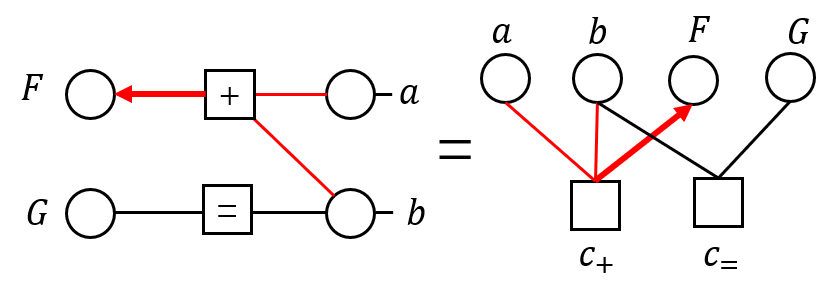}
}
\quad
\subfigure[V2C in polar code]{
\includegraphics[width=0.35\textwidth]{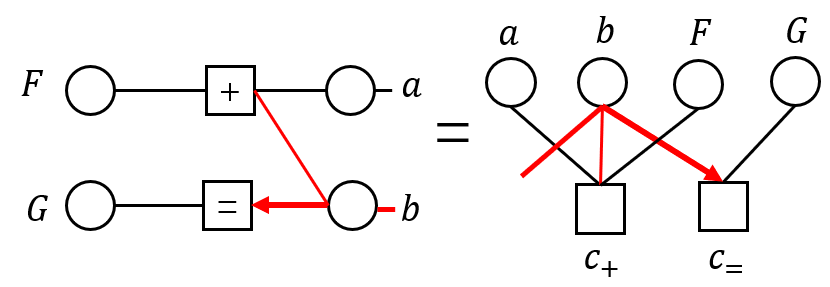}
}
\caption{Four types of message passing in LDPC and polar codes.}
\label{fig1}
\end{figure}

\begin{table}[htbp]
\label{tabCOM}
\begin{center}
\caption{Floating-point computation of each type of message passing under Min-sum algorithm.}
\begin{threeparttable}
\begin{tabular}{|c|c|c|c|}
\hline
         &  Addition & Subtraction & Comparison \\ \hline
LDPC C2V & 0                            & 0           &{$2-\frac{3}{d_c}$}         \\ \hline
LDPC V2C & 1                            & 1           & 0          \\ \hline
Polar C2V  & 0                            & 0           & 1          \\ \hline
Polar V2C  & 1                            & 0           & 0          \\ \hline
\end{tabular}
\begin{tablenotes}
\footnotesize
    \item[1]    {$d_c$ denotes the degree of check node.}
\end{tablenotes}
\end{threeparttable}
\end{center}
\end{table}

Given the above analyses, we find that different types of message passing involve similar amount of computation, which allows us to use NMP as a simplified measure of complexity. {Although} NMP does not capture every detail, it provides a high-level abstraction of the complexity model for message-passing decoding algorithms.

\subsection{Gap to maximum \emph{A Posteriori} (GAP) as the performance metric (y-axis)}

In this subsection, we discuss the performance metrics to track the convergence process as messages are passed along the edges. Many performance metrics can be used for this purpose. Naturally, one would consider BER and BLER. For BER, we can obtain a good approximation based on the average entropy. For BLER, it is only observable at the final decoding stage, and fails to track the effect of early-stage message passing. We instead define a new metric based on Bethe free energy.

{By introducing NMP into density evolution, we can define the average entropy of a factor graph after $t$ messages passed.
\begin{definition}
    The average entropy, $E: \mathcal{X} \times \mathbb{N} \rightarrow \mathbb{R}$, of the number of messages passed (NMP) $t$ and channel messages with $L$-density $\bm{c}=(c_1,\cdots,c_N)$
\begin{equation}
\begin{aligned}
\mathrm{E}(\bm{c}, t) \triangleq \frac{1}{N} \sum_{i} \mathrm{H}\left(c_{i} \circledast \left(\circledast_{\alpha \in N(i)} c_{(\alpha, i)}^{(t)}\right)\right).
\end{aligned}
\end{equation}
\label{defentropy}
\end{definition}
The average entropy $\mathrm{E}(\bm{c}, t)$ can be analytically or numerically obtained. As shown in Fig. \ref{fig.UPPERBOUND} and the following theorem, it accurately matches the BER performance.
\begin{theorem}
    Let $P_{i,t}$ denote the error rate of variable node $i$ after $t$ messages passed, if the factor graph is cycle-free, then we have the average error rate of all variable nodes
    $$
    P_{t} =\frac{1}{N}\sum_iP_{i,t} \leq \frac{1}{2}\mathrm{E}(\bm{c}, t).
    $$
\end{theorem}
\begin{proof}

    Let $x_i$ and $y_i$ denote the channel input and output of variable node $i$ respectively, and $Y_{i,t}$ denote the set of channel outputs of the neighbors connected to $x_i$ after $t$ messages passed.
    The hard decision through the marginal probability ($P(x_i=0)>0.5$ or not) is {bit-wise optimal}. By the Kovalevskij inequality \cite{5625631}, we have
    $$
    P_{i,t} \leq \frac{1}{2}H(x_i\mid y_i, Y_{i,t}) = \frac{1}{2}H\left(c_{i} \circledast \left(\circledast_{\alpha \in N(i)} c_{(\alpha, i)}^{(t)}\right)\right).
    $$
\end{proof}
\begin{figure}[htbp]
    \centering
    \includegraphics[width=0.6\textwidth]{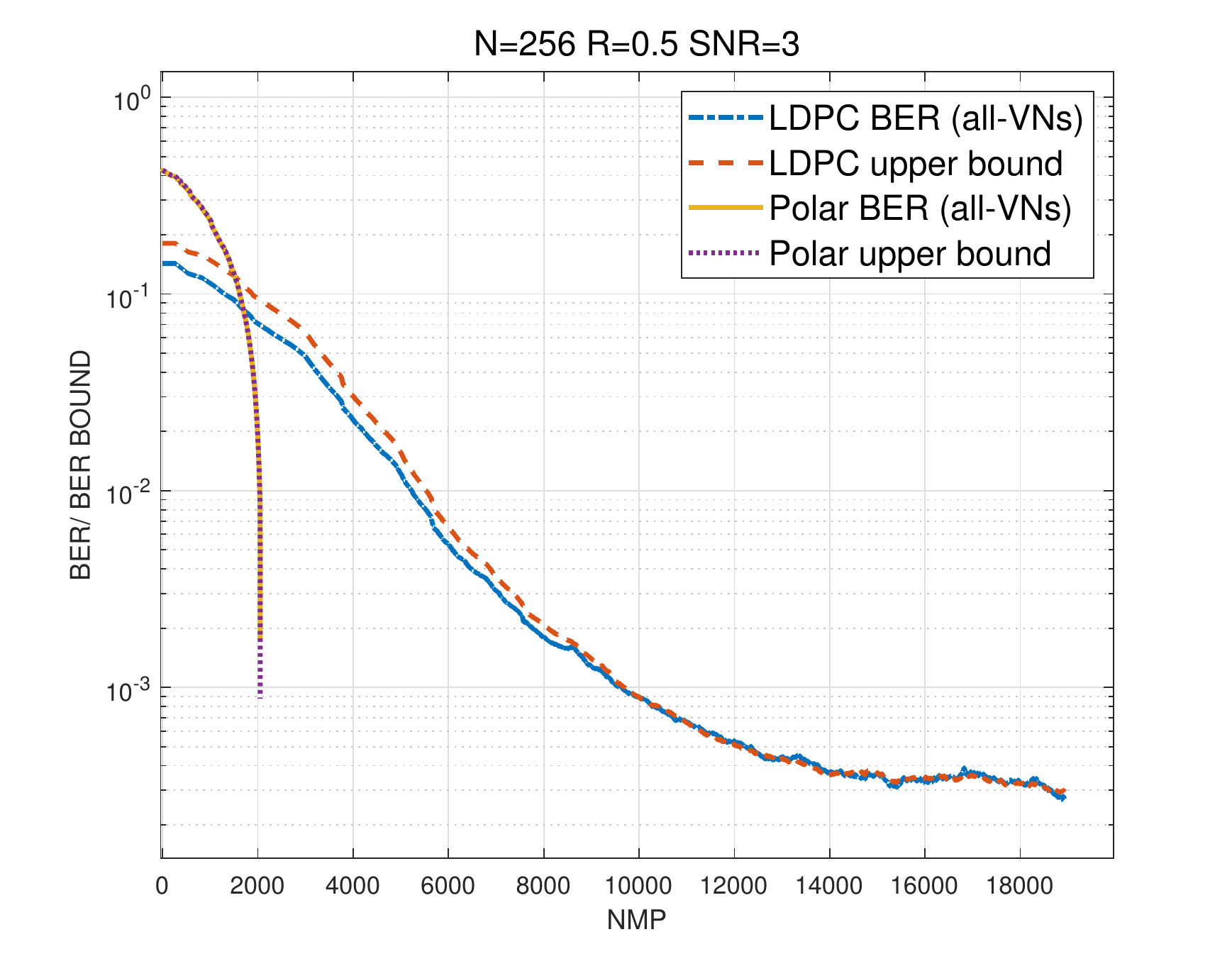}
    \caption{BER evolves during message passing (constructions follow 5G New Radio).}
    \label{fig.UPPERBOUND}
\end{figure}

It should be noted that the above BER is defined over all variable nodes (VNs) in the factor graph, while conventionally a BER is defined over the VNs corresponding to code bits only. It leads to the phenomenon that the initial values vary over different coding schemes, which can be observed in Fig.~\ref{fig.UPPERBOUND}.
}

An alternative metric is the statistical distance between an intermediate estimate and the maximum \emph{a posteriori} estimate of the received codeword. Similar metrics like Bethe free energy have been applied to prove the threshold saturation for spatially-coupled LDPC codes \cite{6912949}. The potential functional in \cite{6912949} {can be used to qualitatively determine} whether BP decoding will converge to a fixed point, but this work provides a quantitative analysis to track the behavior of BP decoding with each message passed  until convergence. Rather than {adopting} the original definition of Bethe free energy or averaging over the code ensemble \cite{6912949}, we consider a potential functional that results from averaging Bethe free energy over the channel realizations. Since we no longer average over the code ensemble, we are able to consider the decoding process of a specific code or even a specific factor graph. This will be demonstrated shortly.

First we {propose our GAP potential functional by modifying the original definition of Bethe free energy.} Recall the expectation of Bethe free energy with respect to the channel output
$$
\begin{aligned}
E_{\bm{y}}(F_{*}(\underline{p})) &=-\{\sum_{i} \mathrm{H}\left(c_{i} \circledast \left(\circledast_{\alpha \in N(i)} c_{(\alpha, i)}\right)\right) \\
&+\sum_{\alpha} \sum_{i \in N(\alpha)} \mathrm{H}\left(c_{(i, \alpha)}\right)-\sum_{\alpha} \mathrm{H}\left(\boxast_{i \in N(\alpha)} \mathrm{c}_{(i, \alpha)}\right) \\
&-\sum_{(i, \alpha)} \mathrm{H}\left(\mathrm{c}_{(i, \alpha)} \circledast c_{(\alpha, i)}\right)\}+\sum_{i}H(c_i).
\end{aligned}
$$
This formula focuses on $c_{(\alpha, i)}$ and $c_{(i, \alpha)}$, the convergence of $L$-densities on the BP fixed point. To characterize the decoding process, we find that even if $\underline{p}$ does not satisfy the fixed-point equation, $F_*(\underline{p})$ is still computable and its stationary point is a BP fixed point. Hence we can expand the domain of the variables $\mathrm{c}_{(i, \alpha)}$ and $\mathrm{c}_{(\alpha, i)}$ in $E_y(F_{*}(\underline{p}))$ from the BP fixed point to the whole decoding process by adding the number of messages passed as a new variable, and remove the terms $\sum_{i}H(c_i)$ which do not change with the decoding process. Based on such a generalization, we define a potential functional to track the message passing process.

\begin{definition}
The gap to maximum a posteriori (GAP) is the potential functional, $U: \mathcal{X} \times \mathbb{N} \rightarrow \mathbb{R}$, of the number of messages passed $t$ and channel messages with $L$-density $\bm{c}=(c_1,\cdots,c_N)$

\begin{equation}
\begin{aligned}
\mathrm{U}(\bm{c}, t) & \triangleq -\frac{1}{N} \{\sum_{i} \mathrm{H}\left(c_{i} \circledast \left(\circledast_{\alpha \in N(i)} c_{(\alpha, i)}^{(t)}\right)\right) +\sum_{\alpha} \sum_{i \in N(\alpha)} \mathrm{H}\left(c_{(i, \alpha)}^{(t)}\right)\\
&-\sum_{\alpha} \mathrm{H}\left(\boxast_{i \in N(\alpha)} \mathrm{c}_{(i, \alpha)}^{(t)}\right) -\sum_{(i, \alpha)} \mathrm{H}\left(\mathrm{c}_{(i, \alpha)}^{(t)} \circledast c_{(\alpha, i)}^{(t)}\right)\}.
\end{aligned}
\label{equ100}
\end{equation}
\label{def2}
\end{definition}
The initial value and update rule are given as follows.

$$c_{(\alpha, i)}^{(0)}= c_{(i, \alpha)}^{(0)}=\Delta_0,$$
$$c_{(\alpha, i)}^{(t)}=\boxast_{j \in N(\alpha)\backslash i} \ c_{(j, \alpha)}^{(t-1)},$$
and
$$c_{(i, \alpha)}^{(t)}=c_{i} \circledast \left(\circledast_{h \in N(i)\backslash \alpha } c_{(h, i)}^{(t-1)}\right).$$
Specifically, when $i$ represents an information bit or frozen bit in polar codes, we have $c_{(i, \alpha)}^{(t)}=\Delta_{\infty}.$

{
\begin{remark}
Note that the average entropy in Definition \ref{defentropy} appears as part of GAP. The four terms in GAP denote the entropy of the variable nodes, the edges, the check nodes, and the entropy of variable nodes at time $(t-1)$, respectively. Different from the average entropy, the initial value of GAP only depends on the code rate and the channel. The proof of this result will be given in the next subsection.
\end{remark}
}
\begin{remark}
In Definition \ref{def2}, we assume that the factor graph $G$ and scheduling policy $S$ are fixed. Strictly, $\mathrm{U}(\bm{c}, t)$ should be written as $\mathrm{U}(\bm{c}, t, G, S)$, otherwise we do not have $\mathrm{U}(\bm{c}, t)=\mathrm{U}(\bm{c}, t^{'})$ when $t=t^{'}$. The factor graph is associated with the coding scheme, and the scheduling policy can be described as follows.
\begin{definitionNoParens} [(Scheduling policy)]
If we label each directed edge in a factor graph from $1$ to $k$, then a scheduling policy can be denoted by a sequence $S=\{s_1,\cdots, s_p\}$, where $p \in \mathbb{N}$ and $s_i \in [1,k]~(1\leq i\leq p)$.
\end{definitionNoParens}
\end{remark}

The scheduling policy $S$ indicates which edge should be updated at each step. Compared with the description of layered BP or other previous studies, this definition gives us a better characterization of scheduling policy. By viewing decoding as a dynamic process, we can zoom in to observe the differences of $\mathrm{U}(\bm{c}, t)$ between factor graphs and the efficiency of scheduling policies. The comparison results can be found in Sections IV and V.



\subsection{Monotonicity and boundedness}

Next, we will discuss the monotonicity and boundedness of $\mathrm{U}(\bm{c},t)$ and reveal its relationship to message passing. Consider the communication over a BMS channel under the all-zero codeword assumption. Let $X=\{x_1,\cdots,x_N\} \in \{0,1\}^N$ and $Y=\{y_1,\cdots, y_{N}\} \in \mathbb{R}^N$ denote the input and output of the BMS channel respectively, $\bm{c}=(c_1,\cdots,c_N)$ be the $L$-density of the output, and $R$ be the code rate.


\begin{theorem}
Assume that the fixed factor graph $G$ is cycle-free, then for any scheduling policy $S$ and $t \geq 1$, $\mathrm{U}(\bm{c}, t)$ is monotonic decreasing with $t$. The upper {limit} of $\mathrm{U}(\bm{c}, t)$ is $1-R-\frac{1}{N}\sum_{i=1}^{N}H(c_i)$ and the lower {limit} is $-\frac{1}{N}H(X|Y)$.
\label{thm100}
\end{theorem}

The following propositions and lemmas will accomplish the proof of Theorem~\ref{thm100} and help us to understand the meaning of $\mathrm{U}(\bm{c}, t)$ in the decoding process.

\begin{propositionNoParens}[(Initial value)]
The initial value $\mathrm{U}(\bm{c},0)$ of a fixed factor graph is $1-R-\frac{1}{N}\sum_{i}H(c_i)$.
\label{prop1}
\end{propositionNoParens}

This proposition holds clearly by substituting $c_{(\alpha, i)}^{(0)}= c_{(i, \alpha)}^{(0)}=\Delta_0$ into formula (\ref{equ100}). We can find that the initial value is negatively correlated with channel entropy. A negative initial value indicates the code rate exceeds channel capacity.

\begin{lemma}
Consider using a message-passing algorithm in a cycle-free factor graph. For any $t \geq 1$, the $t$-th massage passing (C2V or V2C) satisfies $c_{\alpha, i}^{(t-1)} \succeq c_{\alpha, i}^{(t)}$ or $c_{i, \alpha}^{(t-1)} \succeq c_{i, \alpha}^{(t)}$.
\label{lemma1}
\end{lemma}

\begin{proof}
We use the induction method. For $t=1$, if the first message passing is V2C, then we have $c_{i, \alpha}^{(1)} = c_i \preceq c_{i, \alpha}^{(0)}$; if the first message passing is C2V, then we have $c_{\alpha, i}^{(1)} = c_{\alpha, i}^{(0)}$. We assume that this conclusion holds if $t \leq k$ and consider $t = k+1$. If $i$ is associated with an information bit or frozen bit in the polar code factor graph, then we get
$$c_{(i, \alpha)}^{(k+1)}=\Delta_{\infty} \preceq c_{(i, \alpha)}^{(k)},$$
else
$$
\begin{aligned}
c_{i, \alpha}^{(k+1)} &= c_i \circledast \left(\circledast_{h \in N(i)\backslash \alpha} \ c_{(h, i)}^{(k)}\right) \\
&\preceq c_i \circledast \left(\circledast_{h \in N(i)\backslash \alpha} \ c_{(h, i)}^{(k-1)}\right) \\
&= c_{i, \alpha}^{(k)}.
\end{aligned}
$$

We can similarly get $c_{\alpha, i}^{(k+1)} \preceq c_{\alpha, i}^{(k)}$ which completes our proof.
\end{proof}


\begin{propositionNoParens}[(Monotonic decreasing)]
For any scheduling policy $S$ and $t \geq 1$, the $t$-th message passing (C2V or V2C) satisfies $U(\bm{c},t) \leq U(\bm{c},t-1)$.
\label{thm1}
\end{propositionNoParens}

\begin{proof}
We discuss the two types of message passing separately. If the $t$-th message passing is C2V, that is, there is a $c_{\alpha, i}^{(t-1)}$ updated to $c_{\alpha, i}^{(t)}$, then we only consider the change related to the edge $(\alpha,i)$. Notice that the message update only affects the first and the last terms of $U(\bm{c},t-1)$, hence we have
$$
\begin{aligned}
&NU(\bm{c},t)-NU(\bm{c},t-1)\\
&=-H\left(c_{\alpha, i}^{(t)} \circledast c_{i} \circledast \left(\circledast_{h \in N(i) \backslash \mathrm{\alpha}} c_{h, i}^{(t-1)}\right) \right)\\
&+H\left(c_{\alpha, i}^{(t-1)} \circledast c_{i} \circledast \left(\circledast_{h \in N(i) \backslash \mathrm{\alpha}} c_{h, i}^{(t-1)}\right) \right)\\
&+H\left(c_{\alpha, i}^{(t)} \circledast c_{i,\alpha}^{(t-1)}\right) -H\left(c_{\alpha, i}^{(t-1)} \circledast c_{i,\alpha}^{(t-1)}\right)\\
&=H\left((c_{\alpha, i}^{(t-1)}-c_{\alpha, i}^{(t)})\circledast c_{i} \circledast \left(\circledast_{h \in N(i) \backslash \mathrm{\alpha}} c_{h, i}^{(t-1)}\right) \right)\\
&-H\left((c_{\alpha, i}^{(t-1)}-c_{\alpha, i}^{(t)}) \circledast c_{i,\alpha}^{(t-1)}\right)\\
&=H\left((c_{\alpha, i}^{(t-1)}-c_{\alpha, i}^{(t)})\circledast (c_{i} \circledast \left(\circledast_{h \in N(i) \backslash \mathrm{\alpha}} c_{h, i}^{(t-1)}\right)-c_{i,\alpha}^{(t-1)}) \right).
\end{aligned}
$$
By Lemma~\ref{lemma1}, we have $c_{\alpha, i}^{(t-1)} \succeq c_{\alpha, i}^{(t)}$. Since $H(x_1 \circledast x_2) \geq H(x_3 \circledast x_2)$ when $x_1 \succeq x_3$, we have
$c_{i, \alpha}^{(t-1)} = c_{i} \circledast \left(\circledast_{h \in N(i) \backslash \mathrm{\alpha}} c_{h, i}^{(t-2)}\right) \succeq c_{i} \circledast \left(\circledast_{h \in N(i) \backslash \mathrm{\alpha}} c_{h, i}^{(t-1)}\right)$, hence $U(c,t)-U(c,t-1) \leq 0$.

If the $t$-th message passing is V2C, using $H(x_1 \boxast x_2)=H(x_1)+H(x_2)-H(x_1 \circledast x_2)$, then we can similarly get

$$
\begin{aligned}
&U(\bm{c},t)-U(\bm{c},t-1)\\
&=\frac{1}{N}H\left((c_{i,\alpha}^{(t-1)}-c_{i,\alpha}^{(t)})\circledast \left(\left(\boxast_{j \in N(\alpha) \backslash i} c_{j, \alpha}^{(t-1)}\right)-c_{\alpha,i}^{(t-1)}\right) \right)\\
&\leq 0.
\end{aligned}
$$

\end{proof}

We can get a sufficient and necessary condition for $U(\bm{c},t)$ to be strictly decreasing.
\begin{corollary}
If the $t$-th message passing is C2V, then $U(\bm{c},t) < U(\bm{c},t-1)$ if and only if
$$c_{\alpha, i}^{(t-1)} \succ c_{\alpha, i}^{(t)} \text{ and } c_{i, \alpha}^{(t-1)} \succ c_{i} \circledast \left(\circledast_{h \in N(i) \backslash \mathrm{\alpha}} c_{h, i}^{(t-1)}\right).$$

If the $t$-th message passing is V2C, then $U(\bm{c}, t) < U(\bm{c}, t-1)$ if and only if
$$c_{i, \alpha}^{(t-1)}\succ c_{i, \alpha}^{(t)} \text{ and } c_{\alpha, i}^{(t-1)}\succ \left(\mathbb{\boxast}_{j \in N(\alpha) \backslash \mathrm{i}} c_{j, \alpha}^{(t-1)}\right).$$
\label{coro1}
\end{corollary}
Corollary~\ref{coro1} implies that ``better'' messages may not lead to an instant decrease of the potential functional. The ``delay'' of message update is necessary. The potential functional of the SC algorithm in the polar code factor graph characterizes this phenomenon.
\begin{corollary}
In the SC decoding algorithm, $U(\bm{c},t)$ decreases if and only if the updated edge is incident with a frozen bit node.
\label{prop5}
\end{corollary}

Since only frozen bits can offer messages with
$$c_{i, \alpha}^{(t-1)} \succ c_{i} \circledast \left(\circledast_{h \in N(i) \backslash \mathrm{\alpha}} c_{h, i}^{(t-1)}\right),$$
by Corollary~\ref{coro1}, we can get this property which leads to many interesting observations to be shown in the next section.

Then we introduce the conditional entropy of tree code, {which corresponds to the lower limit of} $U(\bm{c},t)$.
\begin{lemmaNoParens} [\cite{6589171}]
If the factor graph is cycle-free, then
$$
\begin{aligned}
H(X|Y)&=\sum_{i} \mathrm{H}\left(c_{i} \circledast \left(\circledast_{\alpha \in N(i)} c_{(\alpha, i)}\right)\right) \\
&+\sum_{\alpha} \sum_{i \in N(\alpha)} \mathrm{H}\left(c_{(i, \alpha)}\right)-\sum_{\alpha} \mathrm{H}\left(\boxast_{i \in N(\alpha)} \mathrm{c}_{(i, \alpha)}\right) \\
&-\sum_{(i, \alpha)} \mathrm{H}\left(\mathrm{c}_{(i, \alpha)} \circledast c_{(\alpha, i)}\right).
\end{aligned}
$$
\label{lemma6}
\end{lemmaNoParens}

\begin{propositionNoParens} [(Lower limit)]
For fixed cycle-free factor graph $G$ and any scheduling policy, we have $$U(\bm{c},t)\geq -\frac{1}{N}H(X|Y)\geq -P_e -\frac{H_b(P_e)}{N},$$ where $P_e$ is the block error rate and $H_b$ is the binary entropy function.
\label{thm7}
\end{propositionNoParens}


\begin{proof}
From Corollary~\ref{coro1}, we can find that $U(\bm{c},t)$ at the BP fixed point is minimal. At the same time, by Lemma~\ref{lemma6}, the minimal value is $-\frac{1}{N}H(X|Y)$. The inequality between the BLER and entropy follows from the Fano inequality \cite{5625631}. 
\end{proof}


{

 For the codes which can deliver reliable communications ($P_e$ is {close to zero}), we have {$U(\bm{c},t) \geq -\frac{1}{N}H(X|Y) \approx 0$, so} the GAP curves of these codes have the same lower {limit}. As discussed in Chapter 15 of \cite{Montanari2009Information}, the equality $U(\bm{c},t) = -\frac{1}{N}H(X|Y)$ holds if and only if the BP decoder achieves the MAP decoding result. Hence, under the assumption of reliable communication, $U(\bm{c},t) > 0$ implies that there is a gap between BP and MAP decoding. In the subsequent sections, we present simulation results to show that, for practical finite-length codes, GAP can be used to compare the decoding efficiencies of various decoding algorithms.
}

\section{Decoding Efficiency analysis of several algorithms}

In this section, we plot the ``number of messages passed (NMP) versus the Gap to maximum \emph{A Posteriori} (GAP)'' curves of several algorithms through simulations and GA analysis for the AWGN channel. The LDPC and polar codes construction follow 5G New Radio. We will introduce these ``NMP-GAP'' curves from three aspects: complexity, slope, and analytical method.

In the following {sequence}, we compare six decoding algorithms for polar and LDPC codes in terms of decoding efficiency:
\begin{enumerate}
\item Flooding BP decoding for LDPC codes;
\item Layered BP decoding for LDPC codes;
\item SC decoding for polar codes;
\item Simplified SC (SSC) decoding \cite{6065237} for polar codes;
\item Soft cancellation (SCAN) decoding \cite{6804940} for polar codes;
\item BP decoding \cite{0Polar} for polar codes.
\end{enumerate}

\subsection{Complexity}

This subsection compares the complexities of the six algorithms and discusses the relationship between complexity and code length. Fig.~\ref{fig3} shows the simulation results of the six algorithms. At the same time, Fig.~\ref{fig10} shows the corresponding SNR-BLER curves as a reference. According to Fig.~\ref{fig10}, we observe that these coding schemes have similar (or slightly different) performances, but we {may overlook their drastic complexity differences without} Fig.~\ref{fig3}. A key information missing in SNR-BLER curve is a quantitative measurement of the decoding complexity. For many applications, we are more concerned about such a huge complexity difference than a modest BLER performance difference (for example, $<0.5dB$). With the proposed paradigm, we are able to compare various decoding algorithms from a new perspective.

As shown in Fig.~\ref{fig3}, the SSC and SC algorithms achieve high efficiency through proper scheduling. This explains why Polar SC/SSC can achieve terabit-per-second decoding. SCAN also has relatively low complexity. But the BP algorithm for polar codes, which runs on the same factor graph, has the lowest decoding efficiency among all schemes. However, the BP algorithm for LDPC codes incurs much less complexity than that of polar codes, thanks to the sparse factor graph. Still, both flooding-BP and layered-BP algorithms exhibit lower decoding efficiency than SC decoding.



\begin{figure}[htbp]
    \centering
    \includegraphics[width=0.7\textwidth]{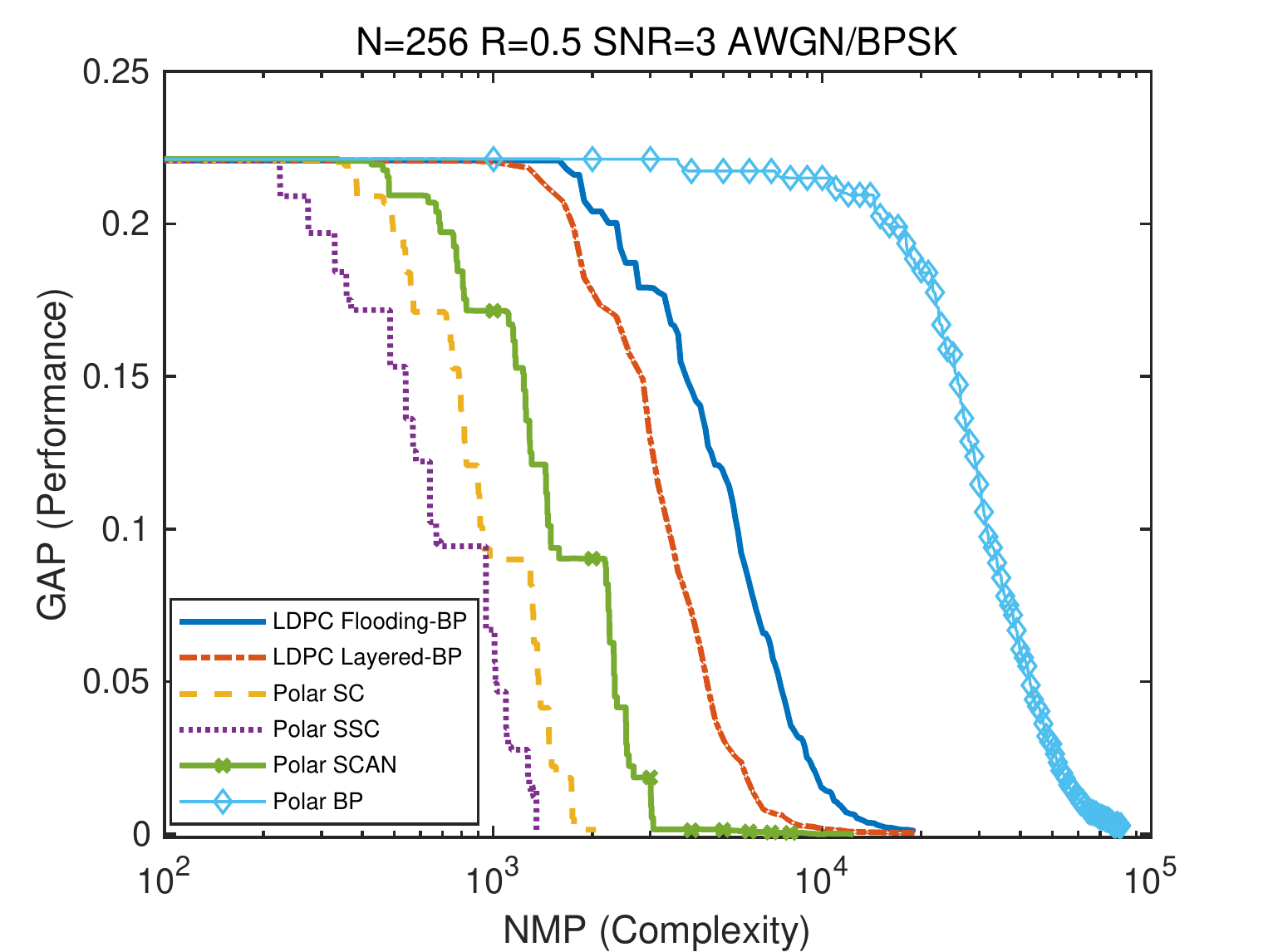}
    \caption{The NMP-GAP curves obtained by simulations. }
    \label{fig3}
\end{figure}

\begin{figure}[htbp]
    \centering
    \includegraphics[width=0.7\textwidth]{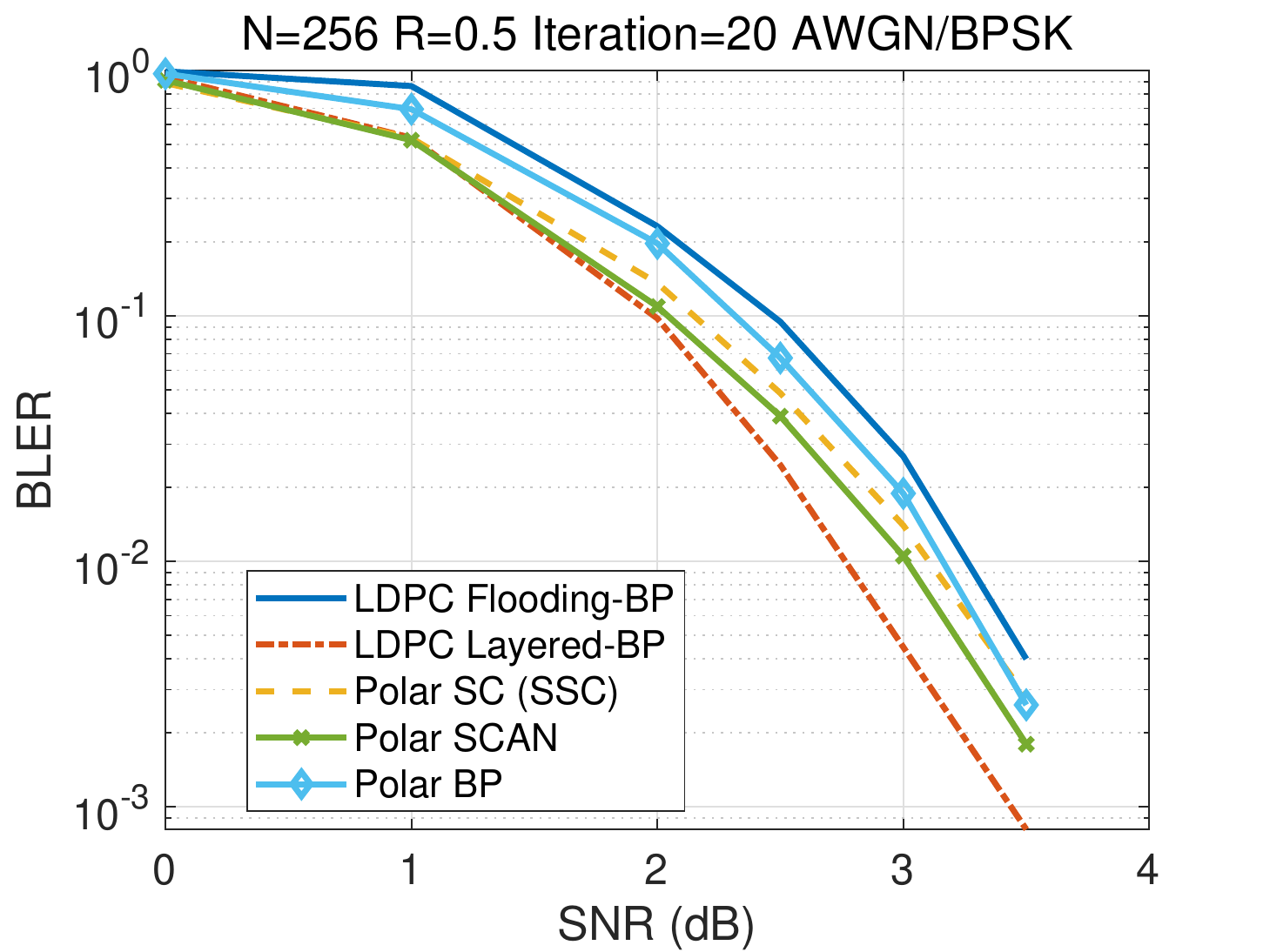}
    \caption{The SNR-BLER curves obtained by simulations. }
    \label{fig10}
\end{figure}

Next, we show that the asymptotic complexity is {occasionally} too coarse to compare decoding algorithms fairly. Given a polar code with length $N$, the asymptotic complexities of the SC, SCAN and polar BP algorithm are $O(N\log {N})$ \cite{6804940, 0Polar}, and that of SSC is $O(N\log {\log {N}})$ \cite{6804939,8010821,9174141}, while the complexities of BP algorithms for LDPC codes are generally considered $O(N)$ {per iteration}. Fig.~\ref{fig5} shows the NMP required for different code lengths. The NMP required in the SC decoder is $N \log N$, and SSC required is less (the actual number depends on the code used). SCAN and polar BP decoder's NMP are $2 \times iter \times N\log N$, where $iter$ is the number of iterations. The NMP in layered-BP \cite{4373433} and flooding-BP decoder is $2d \times iter \times N$, where $d$ is the average degree of variable nodes. In particular, we note that polar BP and SC algorithms have the same asymptotic complexity, but their difference in terms of NMP can be two orders of magnitude. This {analysis} reveals that, for polar codes, BP decoding requires much more messages passed than SC decoding to achieve the same decoding effect and should be considered less computationally efficient. From Fig.~\ref{fig3}, we observe the decoding efficiency of layered BP is higher than flooding BP, as expected. These observations are consistent with the implementation cost in practice but cannot be revealed by asymptotic complexity. The advantage of NMP over asymptotic complexity is that the former provides a higher resolution in reflecting the actual decoding complexity.

The difference between SC and SSC algorithms is also worth {discussing}. In this study, we consider four types of fast-decodable nodes, rate-0 (R0), rate-1 (R1), single parity check (SPC), and repetition (REP) \cite{6065237}. When the decoder traverses to an R0 or R1 node, we only need to make some hard decisions. For an SPC node of length $M$, we need $M-1$ times of comparison; for an REP node, we need $M-1$ times of addition. We count the NMP of an SPC/REP node as $M-1$. The efficiency of SSC has been captured by the ``NMP-GAP'' curves in Fig.~\ref{fig3}. 


\subsection{Slope}
\begin{figure}[htbp]
    \centering
    \includegraphics[width=0.4\textwidth]{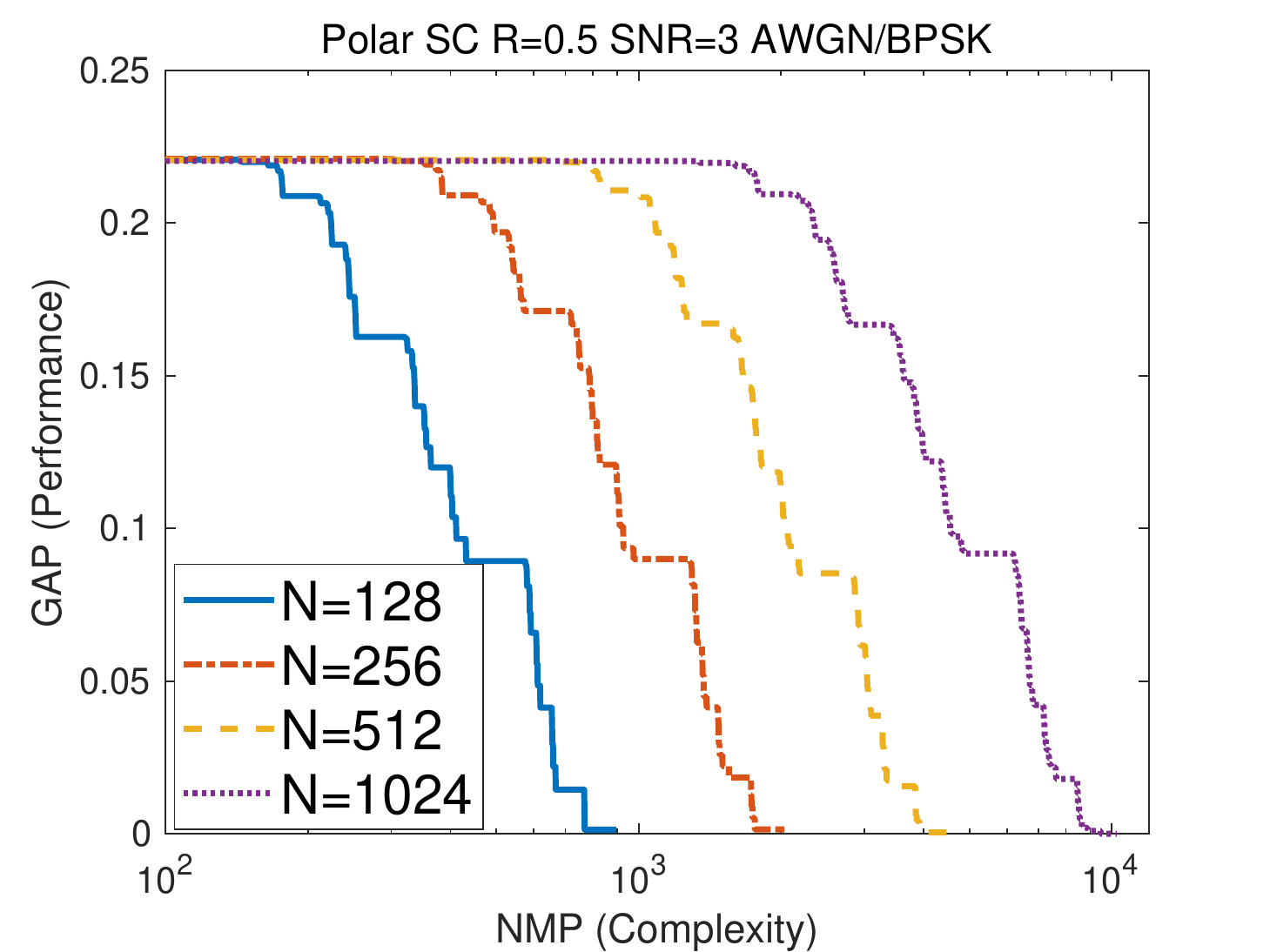}
    \includegraphics[width=0.4\textwidth]{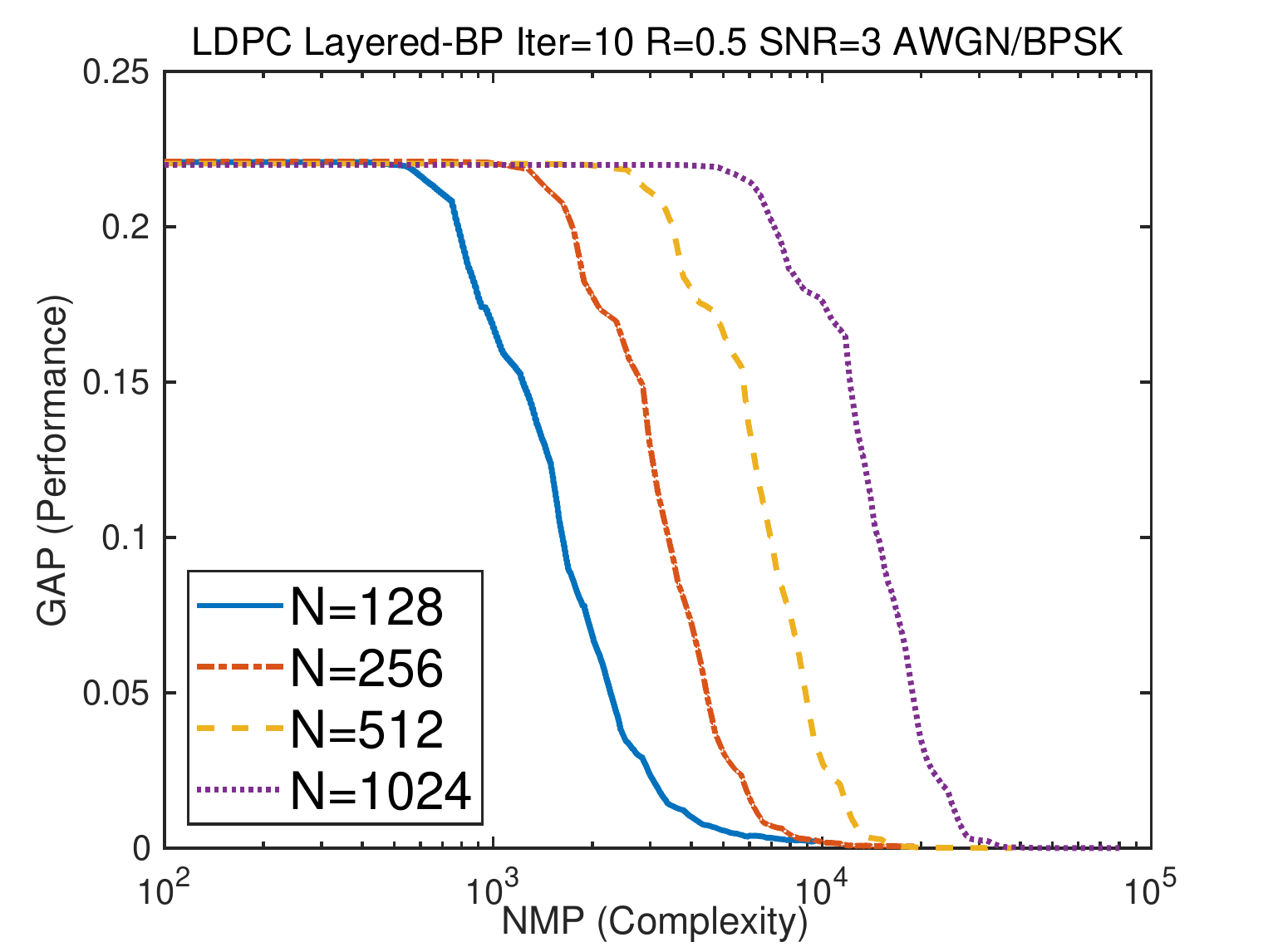}
    \includegraphics[width=0.4\textwidth]{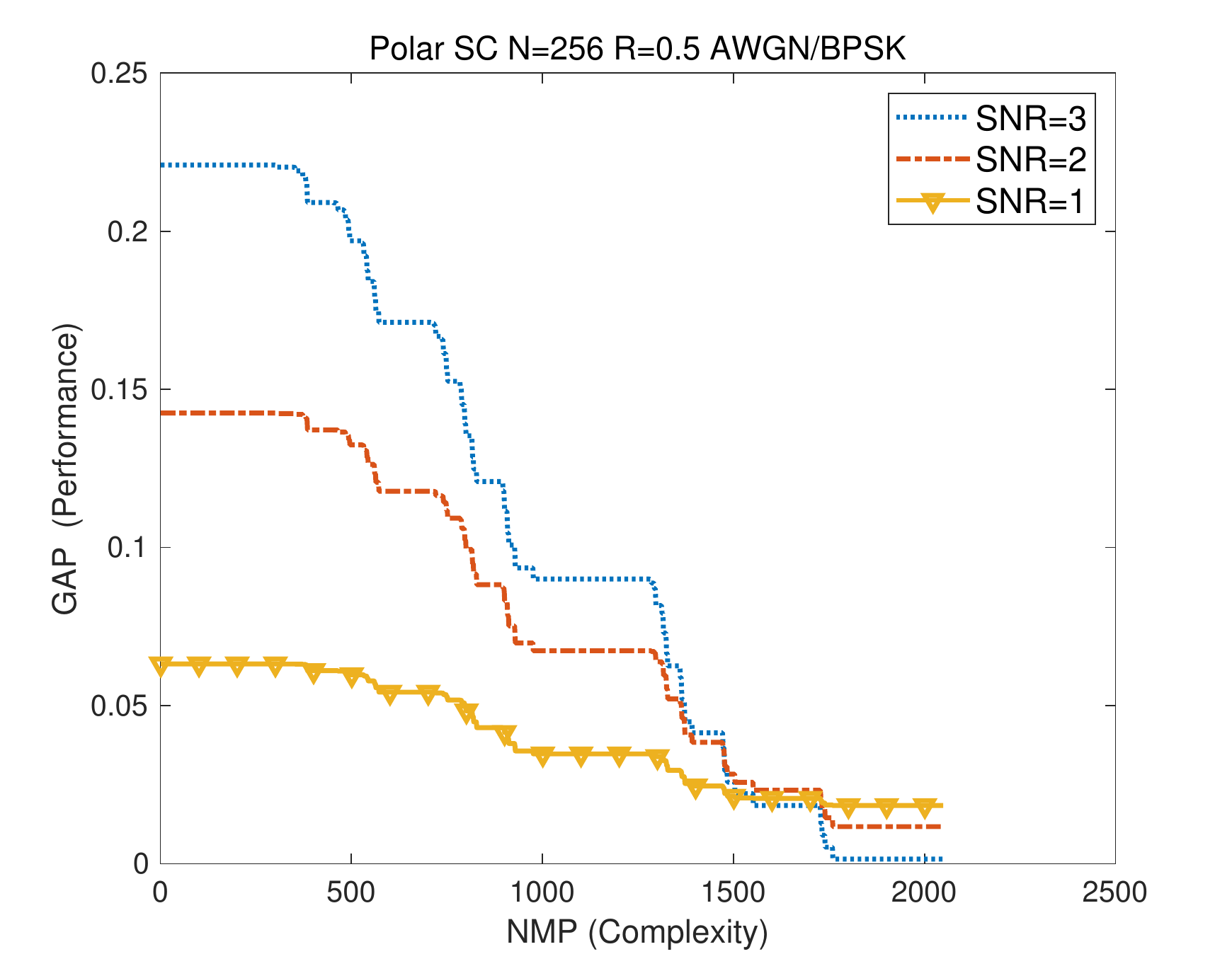}
    \includegraphics[width=0.4\textwidth]{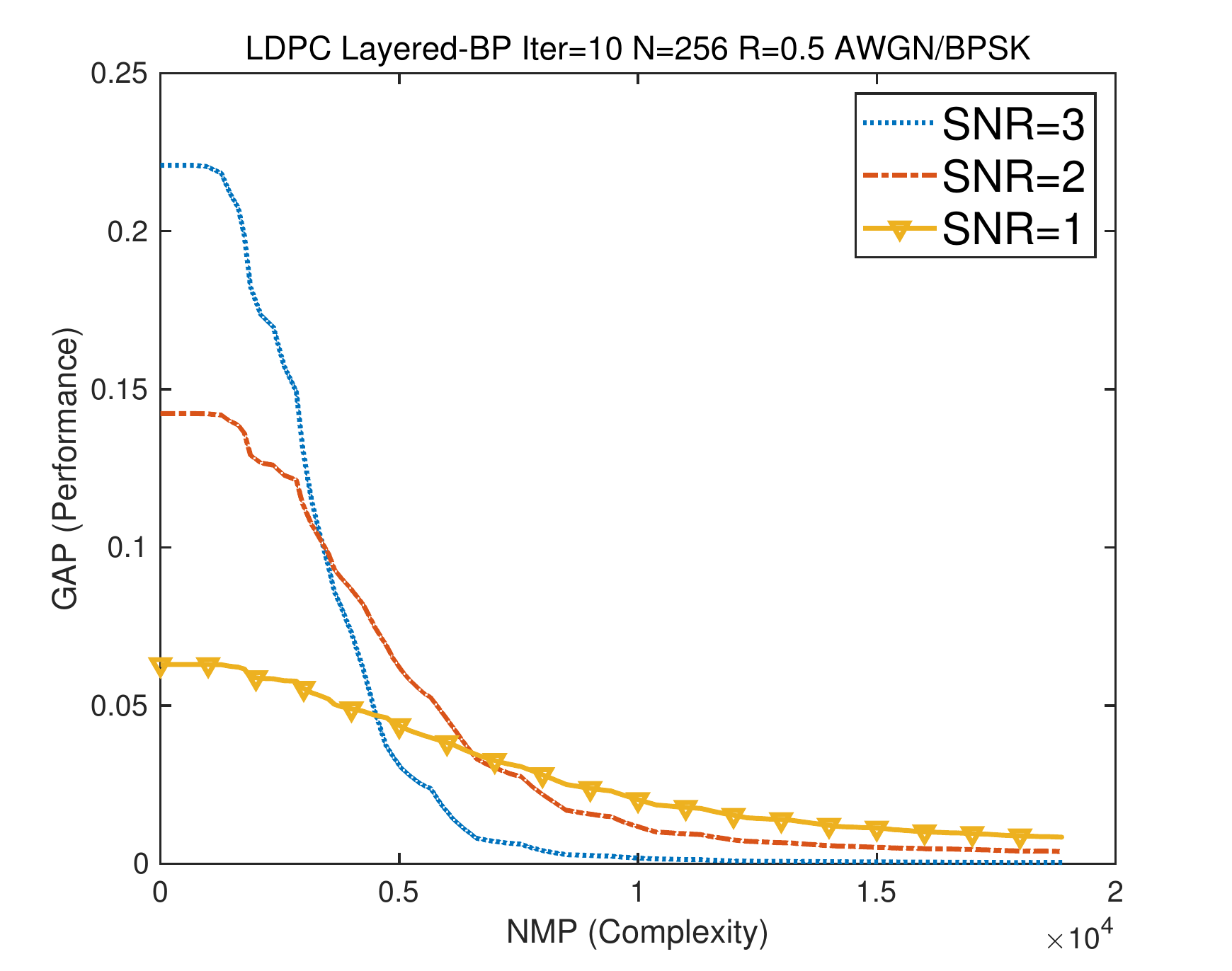}
    \caption{The NMP-GAP curves with different code lengths and different SNRs. }
    \label{fig5}
\end{figure}
The slope of the ``NMP-GAP'' curve is another {indicator of} the decoding efficiency. It is related to decoding algorithm, code length, code rate, and SNR. In Fig.~\ref{fig5}, we find that LDPC codes of shorter lengths have a smaller slope. To some extent, it reflects the inefficiency of such codes. The simulation results also show that the BP algorithm is more efficient at higher SNR because it converges with fewer iterations. {These observations may have been founded in previous studies, but they also demonstrate that our metrics well align with existing knowledge.}

The complexity of polar code does not change with the SNR. The slope of the SC algorithm can be characterized by Theorem~\ref{thm100} and Corollary~\ref{prop5}, which is well verified by the simulation results in Fig.~\ref{fig5}. We can find that the curve {associated with higher SNR} has a larger initial value and is closer to 0 when convergences. %

\subsection{Analytical method}

Alternative to the simulation-based approach, the change of $U(\bm{c},t)$ can also be evaluated analytically by Gaussian approximation \cite{910580}. We initialize the input of decoder by $m_0$ which is normally distributed with mean value $u_0=\frac{2}{\sigma^2}$ and variance $2u_0$, where $\sigma^2$ is the variance of AWGN channel. The Gaussian approximation assumes that the V2C message $m_{i \rightarrow \alpha}$ or C2V message $m_{\alpha \rightarrow i}$ during the decoding process are normally distributed with the mean value $u_{i \rightarrow \alpha}$ and $u_{\alpha \rightarrow i}$, respectively. The update rule for V2C $L$-density's mean value is

\begin{equation}
u_{i \rightarrow \alpha}=u_{0}+\sum_{h \in N(i)\backslash \alpha} u_{h \rightarrow i},
\end{equation}
and the update rule for C2V edge is
\begin{equation}
    \begin{split}
    u_{\alpha \rightarrow i}= \varphi^{-1}\left\{1-\prod_{j \in N(\alpha)\backslash i}\left[1-\varphi\left(u_{j \rightarrow \alpha}\right)\right]\right\},
    \end{split}
\end{equation}
where
\begin{equation}
\varphi(x)=
         \left\{
          \begin{array}{ll}
1-\frac{1}{\sqrt{4 \pi x}} \int_{\mathbb{R}}^{ } \tanh (\frac{u}{2}) e^{-\frac{(u-x)^{2}}{4 x}} d u, & x>0, \\
1, & x=0.
          \end{array}
         \right.
         \end{equation}

\begin{figure}[htbp]
    \centering
    \includegraphics[width=0.4\textwidth]{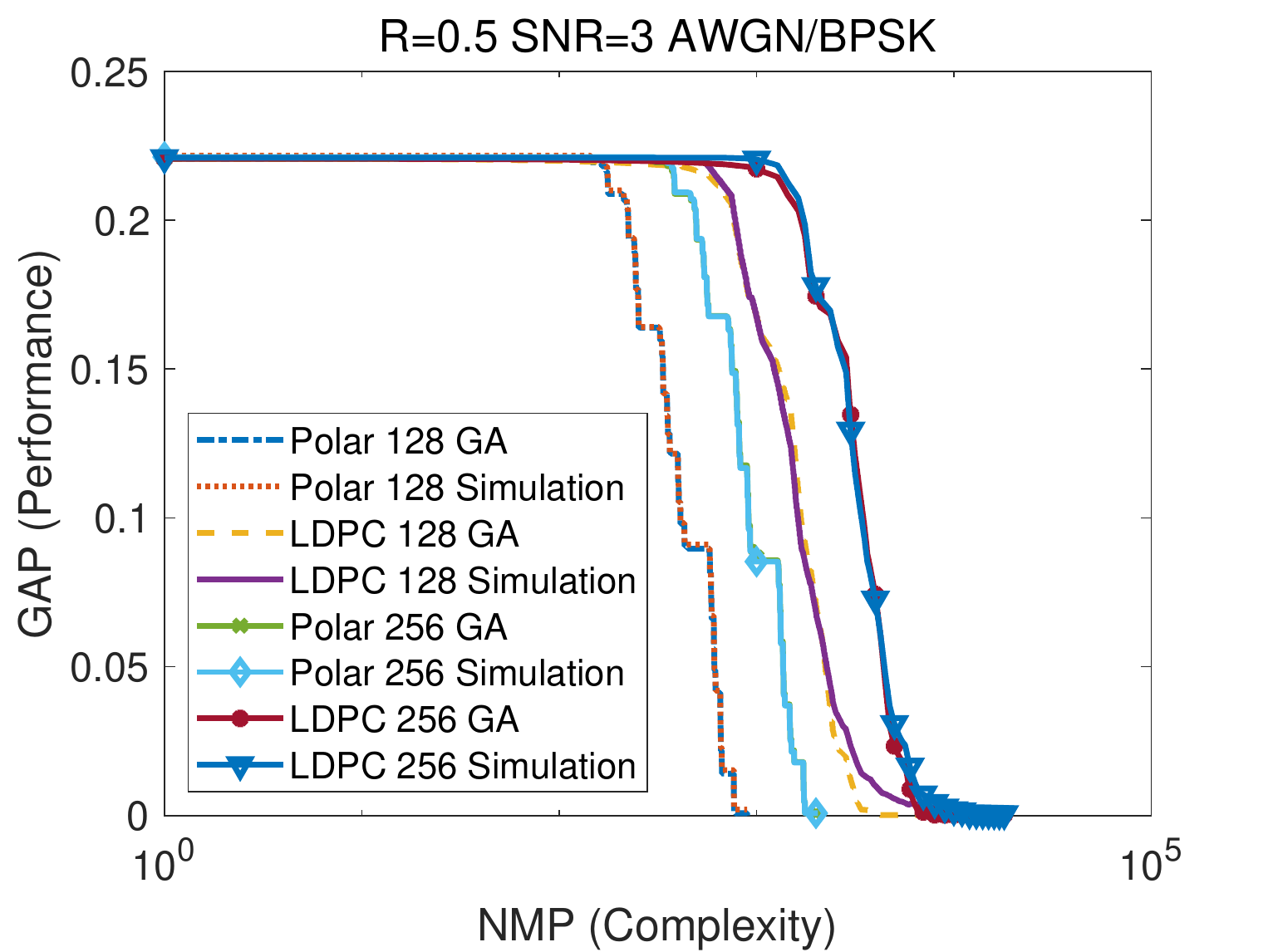}
    \caption{Gaussian approximation and simulation results.}
    \label{fig7}
\end{figure}
We approximate the ``NMP-GAP'' curve by the above updating rules. Fig.~\ref{fig7} shows that there are only very minor deviations from the simulation results. We can employ analytical approaches to speed up the evaluations significantly {which will be used in the next section}.

{
\section{application}
{ In this section, we further illustrate the effectiveness of GAP in measuring efficiency through the application of GAP in the scheduling policy of LDPC codes. 

As discussed in the remark of Definition \ref{def2}, the scheduling policy $S$ can be regarded as a variable of GAP. For a given factor graph $G$, channel with $L$-density $\bm{c}$ and limited complexity $T$, we can formulate the scheduling problem as the following optimization problem,

$$
\min\quad  \tau(\boldsymbol{c},\boldsymbol{S})= \sum_{t=1}^{T}{ U_{\boldsymbol{S}}(\boldsymbol{c},t)}. {} \label{eqn - lp}\\
$$

Based on this optimization problem, we develop a successive-searching BP (SSBP) algorithm to search for a scheduling policy that results in a rapid reduction of the GAP function. This algorithm is based on a local search approach, and more detailed information about SSBP is available in \cite{chang2023optimization}. A comparison of SSBP with the heuristic least-punctured highest-degree (LPHD) scheduling  \cite{9082643} and layered BP is presented in Fig. \ref{figsheduling}. 

\begin{figure}[htbp]
    \centering
    \includegraphics[width=0.8\textwidth]{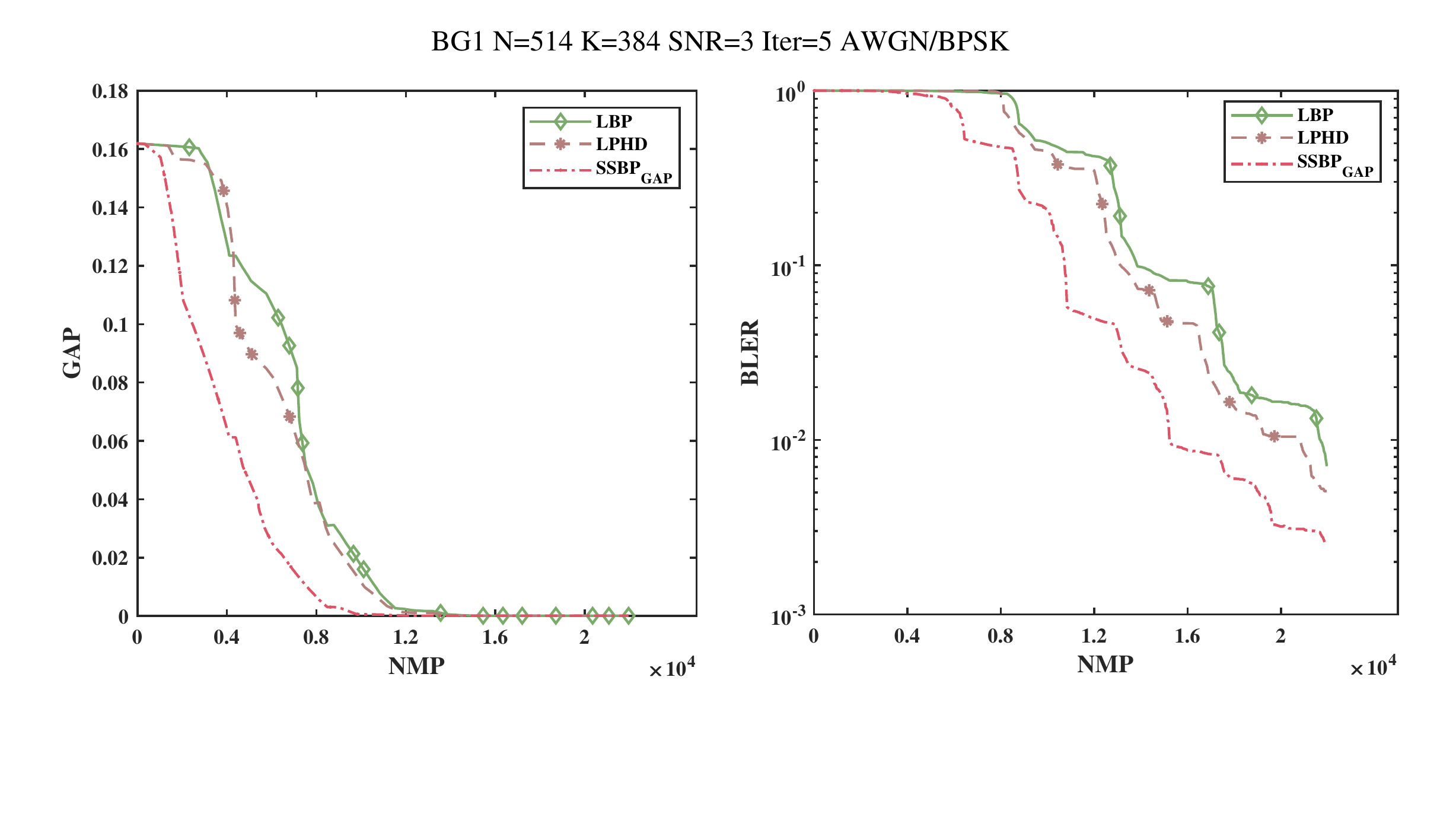}
    \caption{Comparison among scheduling policies in NMP-GAP and NMP-BLER curves.}
    \label{figsheduling}
\end{figure}

Our findings indicate that a decoding algorithm with a faster GAP decrease rate also experiences a faster BLER decrease, suggesting that GAP is a reliable indicator of decoding efficiency. Additionally, our results suggest that GAP may serve as a useful guide for improving decoding efficiency.
}

}

\section{Limitations and future work}
{
There are several weaknesses in the NMP-GAP model that require further study. 

Firstly, GAP as a measure of performance is not applicable to all coding schemes. Although $-\frac{H(X\mid Y)}{N} \rightarrow 0$ for sufficiently large $N$, it is a negative value in the finite length. {Hence when MAP decoding is not error-free, the lower limit of GAP is negative. This means that strictly speaking, GAP is not a distance measure, and a lower GAP does not necessarily imply better performance. Therefore, when applying the NMP-GAP framework to the finite-length regime, the scope should be limited to codes with reasonably good error correction performance (and consequently with the lower limit of GAP close to zero). In these scenarios, GAP is a better performance metric for comparing efficiency than BER and BLER.}


Secondly, NMP in a factor graph as a measure of complexity is also not applicable to all decoding algorithms. For example, it does not apply to several practically important decoders such as SC list (SCL), ordered statistics decoder (OSD). These decoders involve non-message-passing operations, such as path extension, pruning, and sorting, that cannot be analyzed in the message passing framework. Nevertheless, from the perspective of VLSI model, the {amount of} information exchange in a circuit, as a more general form of NMP, is still a fundamental complexity metric. Previous studies by Grover, Blake and Kschischang \cite{6940280,6284015,7083760, 8438511, 8719012, 7862885} showed that regardless of the type of operation used, the key to decoding is still the process of nexus-to-nexus message passing through wires, and the number of messages passed in circuit determines the overall decoding complexity. We may further build upon the message-passing framework to incorporate more decoding algorithms.

Finally, the NMP-GAP model does not provide an optimal message passing scheduling policy in terms of decoding efficiency. In contrast, the spectral efficiency modeled by Shannon's theorem provides insight into the optimal coding scheme. As to decoding efficiency, its theoretical bound for finite length codes still remains unknown. These theoretic results would significantly broaden the scope of information theory research and evaluation methodology for different coding schemes.
}

\section{conclusion}

In this study, we {presented} a model to measure the efficiency of several message-passing decoding algorithms. That is, measuring how many messages {need to be passed} in order to achieve a certain gap from MAP decoding.  Among many well-known decoding algorithms for polar and LDPC codes, we found that SC decoding for polar codes exhibits the best decoding efficiency under the new paradigm, although it is conventionally labeled with moderate performance. We adopted the ``NMP'' to give an implementation-independent complexity measurement and the ``GAP'' to describe the decoding process in a high resolution. The ``NMP-GAP'' method provides insights into the design and evaluation of several decoding algorithms. It prompts us to rethink coding schemes from the “decoding efficiency” perspective, in addition to the “decoding performance” perspective.




%

\bibliographystyle{IEEEtran}
\bibliography{study}
%

\end{document}